%% file: main.tex
\documentclass[a4paper,UKenglish,cleveref, thm-restate]{lipics-v2021}

\pdfoutput=1 
\hideLIPIcs  
\nolinenumbers

\usepackage[textsize=footnotesize,color=green!40]{todonotes}
\usepackage{xspace}
\usepackage{complexity}
\usepackage{subcaption}
\usepackage[textsize=footnotesize,color=green!40]{todonotes}
\usepackage{thmtools}
\usepackage{amsmath}
\usepackage{tikz}
\usetikzlibrary{shapes, calc, graphs, fit, positioning, backgrounds, arrows.meta}
\usepackage[most]{tcolorbox}
\usepackage{adjustbox}

\input{preamble}

\title{Structural parameterizations of Geodetic Set on directed (acyclic) graphs} 

\titlerunning{Parameterizations of Geodetic Set on DAGs} 

\author{Laurent Beaudou}{Université Clermont-Auvergne, CNRS, Mines de Saint-Étienne, Clermont-Auvergne-INP, LIMOS, 63000 Clermont-Ferrand, France\and\url{https://perso.isima.fr/~labeaudo/}}{laurent.beaudou@uca.fr}{https://orcid.org/0000-0003-1959-6855}{}
\author{Florent Foucaud}{Université Clermont Auvergne, CNRS, Mines Saint-Étienne, Clermont Auvergne INP, LIMOS, 63000 Clermont-Ferrand, France\and \url{https://perso.limos.fr/ffoucaud}}{florent.foucaud@uca.fr}{https://orcid.org/0000-0001-8198-693X}{This author received support from the ANR project GRALMECO (ANR-21-CE48-0004), the French government IDEX-ISITE initiative 16-IDEX-0001 (CAP 20-25), the International Research Center ``Innovation Transportation and Production Systems'' of the I-SITE CAP 20-25, and the CNRS IRL ReLaX.}
\author{Lucas Lorieau}{Université Clermont-Auvergne, CNRS, Mines de Saint-Étienne, Clermont-Auvergne-INP, LIMOS, 63000 Clermont-Ferrand, France\and \url{https://perso.limos.fr/~lulorieau/}}{lucas.lorieau@limos.fr}{https://orcid.org/0009-0006-5084-5118}{This author has received financial support from the CNRS through the MITI interdisciplinary programs.}
\author{Prafullkumar Tale}{Indian Institute of Science Education and Research Pune, Pune, India\and\url{https://pptale.github.io/}}{prafullkumar@iiserpune.ac.in}{https://orcid.org/0000-0001-9753-0523}{ARG-MATRICS Grant titled MaPLE (ARGM/2026/0089)}

\authorrunning{L. Beaudou, F. Foucaud, L. Lorieau, P. Tale} 

\Copyright{Laurent Beaudou, Florent Foucaud, Lucas Lorieau, Prafullkumar Tale} 



\EventEditors{John Q. Open and Joan R. Access}
\EventNoEds{2}
\EventLongTitle{42nd Conference on Very Important Topics (CVIT 2016)}
\EventShortTitle{CVIT 2016}
\EventAcronym{CVIT}
\EventYear{2016}
\EventDate{December 24--27, 2016}
\EventLocation{Little Whinging, United Kingdom}
\EventLogo{}
\SeriesVolume{42}
\ArticleNo{23}

\begin{document}

\maketitle

\begin{abstract}
In \gsfull, we are given a (directed) graph and seek a small solution set 
\(S \subseteq V(G)\) such that every vertex lies on a shortest directed path 
between two vertices in \(S\).
While most prior work on \gsfull has focused on undirected graphs, 
in this article we study the problem on directed graphs from the perspective 
of parameterized complexity.

It is known that the problem is \W[2]-hard when parameterized by the solution 
size \(k\), even on directed acyclic graphs (DAGs). 
We investigate 
structural parameterizations of the problem.

Our first result is a kernel of size \(2^{O(\vcn)}\) for \gsfull on general digraphs, 
where \(\vcn\) denotes the vertex cover number of the underlying (undirected) graph. 
This implies an algorithm running in time \(2^{O(\vcn^2)} \cdot n^{O(1)}\).
Furthermore, we prove that, assuming 
the \ETH, the problem does not admit an 
algorithm running in time \(2^{o(\vcn^2)} \cdot n^{O(1)}\). Such a tight quadratic 
exponential lower bound in the parameter is relatively uncommon in parameterized 
complexity. These results generalize earlier work on undirected graphs by Foucaud 
et al.~[STACS 2025], and complements a recent result on directed graph by Foucaud et al.~[CALDAM 2026], 
that showed that the problem is \textsf{para}-\NP-hard 
for the pathwidth and feedback vertex set number of the underlying graph.

Next, we show that on general digraphs, \gsfull admits a natural kernel of size 
\((k\Delta)^{O(\mfd)}\), where \(\Delta\) is the maximum degree and \(\mfd\) denotes 
the reachability diameter of the digraph (a natural analogue of diameter of
undirected graphs). This yields an algorithm running in time 
\((k\Delta)^{O(\mfd \cdot k)}\cdot n^{O(1)}\). We further prove that, assuming the \ETH, 
the problem does not admit an algorithm running in time 
\((k\Delta)^{o(\mfd \cdot k)} \cdot n^{O(1)}\).
Finally, we justify the necessity of combining parameters by establishing the 
following hardness results for \gsfull:
\begin{enumerate}
    \item It is \W[2]-hard parameterized by \(k\), even on digraphs 
    of maximum degree \(3\).
    \item It is \textsf{para}-\NP-hard parameterized by 
    maximum degree and reachability diameter.
\end{enumerate}
One can infer that the problem remains \W[2]-hard when 
parameterized by \(k\), even on graphs of reachability diameter \(3\)
from Ara\'ujo and Arraes~[DAM 2022]. 

All our conditional lower bounds and hardness results hold even when the input 
digraph is restricted to be a DAG.
\end{abstract}


\input{introduction}

\input{prelims}

\input{vertex-cover-MFCS}

\section{Solution size, maximum degree and reachability diameter}
\label{sec-sol-size-max-deg-mfd}
We restate the result proved in this section.
\thmkernel*
\input{kernel-algo}

\input{fpt-lower-bound}

\input{subcubic-solution-size-reduction}

\input{mfd+max-degree-reduction}

\input{conclusion}

\bibliography{references1}



\end{document}

%% file: preamble.tex
\usepackage[dvipsnames]{xcolor}

\newcommand{\gsfull}{\textsc{Directed Geodetic Set}\xspace}

\newcommand{\partsat}{\textsc{3-Partitioned 3-SAT}\xspace}

\newcommand{\vcn}{\textsf{vcn}\xspace}

\newcommand{\yes}{\texttt{yes}\xspace}
\newcommand{\no}{\texttt{no}\xspace}
\newcommand{\TRUE}{\texttt{true}\xspace}
\newcommand{\FALSE}{\texttt{false}\xspace}
\newcommand{\mfd}{\textsf{rdiam}\xspace}
\newcommand{\ETH}{$\mathsf{ETH}$\xspace}

\newcommand{\problemdc}[4]{
	\begin{center}
		\noindent\framebox{\begin{minipage}{#4\textwidth}
				#1\\
				\textbf{Input:} #2\\ 
				\textbf{Task:} #3
		\end{minipage}}
	\end{center}
}

\crefname{observation}{observation}{observations}
\Crefname{observation}{Observation}{Observations}
\crefname{myclaim}{claim}{claims}
\Crefname{myclaim}{Claim}{Claims}

\crefname{itemi}{item}{items}
\Crefname{itemi}{Item}{Items}

\theoremstyle{remark}
\newtheorem{reducrule}[theorem]{Reduction rule}
\newcommand{\Rulename}{Rule}
\newcommand{\Rulesname}{Rules}
\newcommand{\rulename}{rule}
\newcommand{\rulesname}{rules}
\crefname{reducrule}{\rulename}{\rulesname}
\Crefname{reducrule}{\Rulename}{\Rulesname}

%% file: introduction.tex
\section{Introduction}
\label{sec:intro}

We study \textsc{Geodetic Set} on directed graphs (digraphs for short), which we thus call \gsfull. \textsc{Geodetic Set} is a classic network design problem introduced in~\cite{harary1993} based on the notion of convexity for the shortest path metric. See the books~\cite{BOOKaraujo2025,bookGC} for more on graph convexity, for which \textsc{Geodetic Set} is a central representative problem. In \textsc{(Directed) Geodetic Set}, one asks for a minimum-size set $S$ of vertices such that each vertex in the input (di)graph is part of some shortest directed path between two vertices of $S$. \textsc{(Directed) Geodetic Set} models natural network design problems such as selecting hubs in a transportation network~\cite{floCALDAM20,bus}. It is extensively studied in the literature, and serves as one of the key problems in the area of metric-based algorithmic graph theory, together with its close relative \textsc{Metric Dimension}~\cite{BDM24,floICALP24}. \gsfull is formally stated as follows.

\problemdc{\gsfull}{A directed graph $D$ and an 
integer \(k\).}{Determine whether there exists a subset
\(S \subseteq V(D)\) of size at most \(k\) such that 
every vertex in \(V(D)\) lies on some directed shortest path
between two vertices in \(S\).}{0.95}

\textsc{Geodetic Set} was extensively studied on undirected graphs, on which its algorithmic complexity was almost exhaustively determined for standard graph classes~\cite{wellpart,floISAAC20,floCALDAM20,DBLP:journals/tcs/ChakrabortyGR23,DBLP:journals/iandc/ChakrabortyDFGL26,dourado2010,ekim2012,mezzini2018} and under the lenses of parameterized complexity~\cite{floISAAC20,floICALP24,floSTACS25,KK22,T25}, approximation algorithms~\cite{floCALDAM20,DIT21} and enumeration complexity~\cite{BDM24}.

On digraphs, several studies on \gsfull exist regarding upper and lower bounds~\cite{chang2004geodetic,Chartrand2003,chartrand2000geodetic,dong2009upper,Farrugia05,lu2007geodetic}, see also the book~\cite[Chapter 6]{bookGC}. However, \gsfull was only recently studied on directed graphs from an algorithmic perspective. In~\cite{AraujoA22}, \gsfull was shown to be \NP-hard on DAGs (even when the underlying graph is either split, bipartite or co-bipartite) but polynomial-time solvable on oriented cactus graphs (without 2-cycles). (For a digraph, its \emph{underlying graph} is the undirected graph obtained by forgetting the orientations of arcs, possibly deleting multiple edges.) In a very recent paper~\cite{DBLP:conf/caldam/Foucaud26}, \gsfull was shown to be \NP-hard even on DAGs whose underlying graph has small feedback vertex set number, but it is polynomial-time solvable on digraphs whose underlying graph is a tree, and it is \FPT\ on DAGs for the feedback edge set number of the underlying graph.

Our goal is to continue the study of \gsfull and exhibit parameters for which we can obtain \FPT\ algorithms, and prove their optimality under standard complexity assumptions. One natural parameter is the \emph{vertex cover number} $\vcn$ of the underlying graph 
of the input digraph. Note that there exists an \FPT\ 
algorithm for \gsfull with respect to $\vcn$ on undirected 
graphs~\cite{floSTACS25}. Another relevant parameter is the \emph{reachability diameter}: the maximum directed distance between any pair of vertices in the digraph that are reachable. This parameter was recently introduced in~\cite{DBLP:conf/sosa/HaeuplerJS26} (see also~\cite{bals2026revisitingdiameterdirectedgraphs}) and earlier under the name \emph{maximum finite distance}~\cite{DBLP:journals/algorithmica/FoucaudGPS21} (unlike the classical diameter of a digraph, it is always finite). This metric-based parameter is an analogue of the diameter of a connected undirected graph, and greatly restricts the possible distances in the input digraph.

\subparagraph*{Our results.}

We first extend a result of~\cite{floSTACS25} (for \textsc{Geodetic Set} on undirected graphs) to general directed graphs, by giving a kernel of size $2^{O(\vcn)}$, and deriving from it an \FPT\ algorithm. We also prove that, under the Exponential-Time Hypothesis (\ETH), this algorithm is essentially optimal.

\begin{restatable}{theorem}{thmvc}\label{thm-vc-algo}
    \gsfull admits an algorithm running in time $2^{O(\vcn^2)}\cdot n^{O(1)}$. Moreover, assuming the \ETH, it admits no algorithm running in time $2^{o(\vcn^2)}\cdot n^{O(1)}$, even on DAGs.
\end{restatable}

We then give another positive result by showing that if we combine as parameters the solution size, reachability diameter, and maximum degree of the input digraph, we obtain a kernel and an \FPT\ algorithm. Moreover, we show that its running time is essentially optimal under the \ETH.

\begin{restatable}{theorem}{thmkernel}
\label{thm-kernel}
\gsfull admits a kernel of size $(k\Delta)^{O(\mfd)}$, and an algorithm running in time $(k\Delta)^{O(\mfd\cdot k)}\cdot n^{O(1)}$. Moreover, assuming the \ETH, it admits no algorithm running in time $(k\Delta)^{o(\mfd\cdot k)} \cdot n^{O(1)}$, even on DAGs with $\mfd=2$.
\end{restatable}

Unfortunately, there is little hope to generalize the above result by restricting the parameter combination. Indeed, we show in the two next theorems that one cannot omit neither the reachability diameter, nor the solution size, under standard complexity theory assumptions. 

\begin{restatable}{theorem}{thmWhard}
\label{thm-W-hard}
\gsfull remains \W[2]-hard for solution size $k$, even on DAGs with maximum degree~3.
\end{restatable}

Note that \gsfull is trivially in \XP\ for the solution size~$k$, since one can check (in polynomial time), for each vertex subset of size at most $k$, whether it is a valid solution.

\begin{restatable}{theorem}{thmmaxdeg}
\label{thm-maxdeg}
\gsfull remains \NP-complete on DAGs with maximum degree~3 and reachability diameter~14.
\end{restatable}

Finally, it is known that \gsfull is \NP-complete even on DAGs with reachability diameter~3, and the proof actually implies \W[2]-hardness for solution size on such DAGs~\cite{AraujoA22}. However, the produced inputs have large maximum degree. 

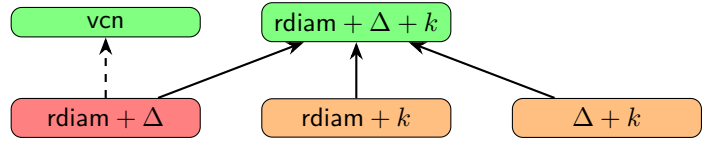
\begin{figure}
    \centering

\begin{tikzpicture}[
  node distance=.8cm,
  every node/.style={draw, rectangle, rounded corners, align=center, minimum width=2.5cm},
  arr/.style={-{Stealth}, thick},
  darr/.style={-{Stealth}, thick, dashed}
]

\node[fill=green!50] (vc) {$\vcn$};
\node[right=of vc,fill=green!50] (mfdDk) {$\mfd+\Delta+k$};



\node[below=of vc, fill=red!50] (mfdD) {$\mfd+\Delta$};
\node[right=of mfdD, fill=orange!50] (mfdk) {$\mfd+k$};
\node[right=of mfdk, fill=orange!50] (Dk) {$\Delta+k$};


\draw[arr] (mfdD) -- (mfdDk);
\draw[arr] (mfdk) -- (mfdDk);
\draw[arr] (Dk) -- (mfdDk);

\draw[darr] (mfdD) -- (vc);
\end{tikzpicture}
    \caption{The parameters studied in this paper, and their relations, where $\vcn$ and $\Delta$ denote the vertex cover number and maximum degree of the underlying undirected graph, $\mfd$ is the reachability diameter, and $k$ is the solution size. An arc indicates that the lower parameter is upper-bounded by a function of the upper parameter. The dashed arc denotes that only \(\mfd\) is directly upper bounded by a function of \(\vcn\), but \(\Delta\) is not. However, with a simple pre-processing on the \gsfull instance, one can obtain
    that \(\Delta \le 2^{O(\vcn)}\). 
    For the top two parameters, \gsfull is \FPT\ (and our algorithms are asymptotically optimal under the \ETH). For the bottom left parameter, \gsfull is para-NP-hard; for the two bottom right parameters, it is $\W[2]$-hard and in \XP.}
    \label{fig:placeholder}
\end{figure}

Our results show that \gsfull remains hard even on very restrictive classes of digraphs. In that sense, it differs from other metric-based parameters such as \textsc{Broadcast Domination}, which is \NP-hard but \FPT\ for solution size on DAGs and also \FPT\ for solution size and maximum degree on general digraphs~\cite{DBLP:journals/algorithmica/FoucaudGPS21}. Note that, interestingly, \gsfull also differs from the related \textsc{Metric Dimension} problem, which is \FPT\ for solution size plus reachability diameter.\footnote{Indeed, in this problem, every vertex needs a distinct distance-vector to the solution vertex set of size $k$ (where the distance between two vertices may be infinite), and there can be at most $k+(\mfd+1)^k$ such distinct vectors in a YES-instance), yielding a trivial kernel of this size.}

\subparagraph*{Outline of the paper.}
We start by a preliminary section in \Cref{sec-prelims}. In \Cref{sec-vcn-para}, we study the tractability of \gsfull with respect to the parameter vertex cover number. \Cref{sec-sol-size-max-deg-mfd,sec-w2-hardness-sol-degree,sec-np-hard-max-deg-mfd} each deals with combinations of solution size, maximum degree and reachability diameter. They contain the proofs of \Cref{thm-kernel,thm-W-hard,thm-maxdeg}, respectively. We end the paper by presenting some open questions in \Cref{sec-conclusion}.

%% file: prelims.tex
\section{Preliminaries}
\label{sec-prelims}
We present general definitions and results that will be used throughout the paper. We refer to the book~\cite{DBLP:books/sp/CyganFKLMPPS15} for terminology and details on parameterized complexity.

A {\it directed graph} (digraph for short) $D$ consists of vertex set $V(D)$ and arc set $A(D)$, where each arc is an ordered pair of two different vertices. An arc from vertex $v$ to vertex $u$ is denoted by $vu$, $v$ is its \emph{tail} and $u$ is its \emph{head}. The \emph{underlying undirected graph} (or simply \emph{underlying graph}) of some digraph $D$ is the graph obtained by removing the orientation of each arc of $D$. An \emph{oriented path} of a digraph $D$ is a subgraph of $D$ whose underlying graph is a path. A \emph{directed path} (or \emph{dipath}) is an oriented path for which all arcs are oriented in the same direction. In the following, except when explicitly stated otherwise, we will simply call a path any directed path. The \emph{distance} from $u$ to $v$ is the length (number of arcs) of a shortest directed path from $u$ to $v$ (if $u=v$, the distance is~$0$). If no such path from $u$ to $v$ exists, then the distance to $u$ to $v$ is defined as $+\infty$. The \emph{reachability diameter} of a digraph $D$ is the largest distance from any vertex to another vertex of $D$ that is not $+\infty$. A digraph is called {\it strongly connected} if any two distinct vertices are connected by two directed paths (one in each direction). A (maximal) \emph{strongly connected component} of digraph $D$ is a maximal vertex subset of $V(D)$ that induces a strongly connected subdigraph of $D$. The \emph{in-neighborhood} of vertex $u$ is denoted by $N^-(u)$ and its {\it out-neighborhood} is denoted by $N^+(u)$. The in-neighborhood of a subset $S$ of $V(D)$ is $N^-(S)=(\cup_{u\in S}N^-(u))\setminus S$, and similarly the out-neighborhood of $S$ is defined. A vertex $x$ is called a {\it source}, if $N^-(x)=\emptyset$, and it is called a {\it sink} if $N^+(x)=\emptyset$. A vertex that is a source or a sink is called \emph{extremal}. For two vertices $u$ and $v$, the set of vertices that lie in some shortest path from $u$ to $v$ is denoted by $I(u,v)$ and for a subset $S$ of $V$, the {\it geodetic closure} of $S$, denoted by $I(S)$, is the set of all vertices which lie in some shortest path between two vertices of $S$. In other words, $I(S)=\cup_{u,v \in S} \big(I(u,v)\cup I(v,u)\big)$. We also say that a vertex $v$ \emph{is covered} by two vertices $u$ and $w$ if $v\in I(u,w)$. In a path $P$ from $u$ to $v$, vertices $u$ and $v$ are called the \emph{tail} and the \emph{head} of $P$, respectively. The vertices of $P$ which are neither the tail nor the head of $P$ are called its \emph{inner vertices}. A {\it directed acyclic graph} (DAG for short) is a digraph which does not contain any directed cycle. A vertex $v$ is \emph{transitive} if, for every in-neighbor $u_1$ and out-neighbor $u_2$ of $v$, either the arc $u_1u_2$ exists, or $u_1=u_2$.\footnote{Note that the latter condition was not present in the definition from~\cite{AraujoA22}, since the authors only considered digraphs without 2-cycles.} A vertex of a digraph is called a \emph{leaf} if, in the underlying undirected graph, it has degree~1. Note that any leaf of a digraph is either a sink, a source, or transitive.

We conclude this section with the following lemma.
\begin{lemma}[\cite{AraujoA22}]\label{lem-sink-sources}
	   In any digraph $D$, every source, sink and transitive vertex of $D$ (in particular, every leaf of $D$) belongs to every geodetic set of $D$.
\end{lemma}

%% file: vertex-cover-MFCS.tex
\section{Vertex cover number parameterization}
\label{sec-vcn-para}

We restate the main theorem of this section.

\thmvc*

\subsection{A kernel and an algorithm}

Two vertices of a digraph $D$ are said to be \emph{false twins} if they share the same incoming and outgoing open neighborhoods. Note that if two false twins share an arc, then one of them has to admit a loop, which is not possible in the considered setting. A set $S$ of vertices is called a \emph{false twin set}, or a \emph{set of false twins}, if all vertices of $S$ are pairwise false twins.
This notion of twins is central in the kernalization algorithm we present in the following. Indeed, since twins share exactly the same neighborhood, they behave exactly in the same manner with respect to the rest of the graph. 
We also need to define the following notion. 

\begin{definition}
The neighborhood of a set $S$ of vertices is said to be \emph{shortcutting} if, for all vertices $v\in S$, $v$ is transitive in $D[(V\setminus (S-v)]$.
\end{definition}

 In particular, if $S$ is an independent set and its neighborhood is shortcutting, then any vertex in $S$ is transitive. The following lemma will be useful. 

\begin{lemma}\label{lemma-accessible}
If a digraph $D$ contains a set $T$ of false twins whose neighborhood is not shortcutting, then any optimal geodetic set of $D$ contains at most four vertices of $T$.
\end{lemma}

\begin{proof}
Let $D$ be a digraph and $T = \{t_1, \dots, t_k\}$ a set of false twins whose neighborhood is not shortcutting. Then, there exist two vertices $x\in N^-(T)$ and $y\in N^+(T)$ so that $(x,y)\not\in A(D)$. Suppose for a contradiction that there exists a minimum geodetic set $S$ of $D$ with at least 5 vertices in $T$. Without loss of generality, say that vertices $t_1, \dots, t_5$ belong to $S$. Define then $S' = \left(S\setminus \{t_1, t_2, t_3\}\right)\cup\{x,y\}$. We claim that $S'$ is a geodetic set of $D$, thus contradicting the fact that $S$ is a minimum-size geodetic set. Let $v\in V(D)\setminus S$ be a vertex that was covered by $S$. If $v$ was covered by a shortest path with endpoints in $S\setminus T$, then it is still covered by the same shortest path. If $v$ was covered by a shortest path from a vertex $w\in S\setminus T$ to another vertex $t\in T\cap S$ (respectively from $t$ to $w$), then it is covered by a shortest path from $w$ to $t_4$ (respectively from $t_4$ to $w$). Finally, if $v$ was covered by a shortest path with both endpoints in $T\cap S$, then it is also covered by a shortest path with endpoints $t_4$ and $t_5$ because there is no arc between them by assumption. It remains to verify that vertices $t_1, t_2$ and $t_3$ are covered. Consider the paths $g_i = (x, t_i, y)$ for some $i\in[3]$. Since by assumption there is no arc between $x$ and $y$ in $D$, $g_i$ is a shortest path for any $i\in[3]$, and vertices $t_1$, $t_2$ and $t_3$ are covered.
\end{proof}

The two following reduction rules remove some vertices of  false twin sets. \Cref{prop-kernel} shows that this is enough to get a kernel for the parameter \vcn. 

\begin{reducrule}\label{rule-shortcutting}
Let $(D,k)$ be an instance of \gsfull and $T$ be a set of at least three  false twins in $D$. If their neighborhood is shortcutting, then delete one vertex of $T$ and decrease $k$ by one.
\end{reducrule}

\begin{reducrule}\label{rule-not-shortcutting}
Let $(D,k)$ be an instance of \gsfull and $T$ be a set of at least six  false twins in $D$. If their neighborhood is not shortcutting, then delete one vertex of $T$ without decreasing $k$.
\end{reducrule}

\begin{proposition}\label{prop-kernel}
The algorithm applying exhaustively \Cref{rule-not-shortcutting,rule-shortcutting} in a sequential manner on an instance $(D,k)$ returns an equivalent instance $(D',k')$ with $|V(D')| = 2^{O(\vcn)}$.
\end{proposition}

\begin{proof}
Let $(D,k)$ be an instance of \gsfull. Consider \Cref{rule-shortcutting} and assume that $D$ contains three vertices $u$, $v$ and $w$ that are  false twins with a shortcutting neighborhood. Suppose that $D$ contains a geodetic set $S$ of size (at most) $k$ (\emph{i.e.} $(D,k)$ is a \yes-instance). Vertices $u$, $v$ and $w$ are transitive since their neighborhood is shortcutting, 
and hence they belong to $S$. Apply \Cref{rule-shortcutting} by deleting $u$ from $D$ to create a new digraph $D'$, and define $S' = S\setminus\{u\}$. Any vertex that was covered by a shortest path in $D$ with terminal vertices $u$ and $z\in S$ is covered by a shortest path in $D'$ with terminal vertices $v$ and $z$ (because $u$ and $v$ are twins in $D$), thus $(D', k-1)$ is a \yes-instance for \gsfull. Conversely, if $D'$ admits a geodetic set $S'$ of size at most $k-1$, then $S'\cup\{u\}$ is a geodetic set of size at most $k$ in $D$. Applying \Cref{rule-shortcutting} thus creates an equivalent instance each time it is used.

Now, consider \Cref{rule-not-shortcutting} and assume that there exists a set $V' = \{v_1, \dots, v_i\}$ of at least 6 vertices that is a strict false twin set and whose neighborhood is not shortcutting. Suppose first that $D$ admits a minimum-size geodetic set $S$ of size (at most) $k$. By \Cref{lemma-accessible}, $S$ contains at most four vertices of $V'$. Without loss of generality, suppose $v_1, \dots, v_4 \in S$. Delete $v_6$ from $D$ to construct a new digraph $D'$. Then, it is clear that $S$ is a geodetic set of $D'$ and $(D', k)$ is a \yes-instance for \gsfull. Conversely, suppose that $D'$ admits a geodetic set of size at most $k$. Again by \Cref{lemma-accessible}, we can assume that $S$ contains at most four vertices from $\{v_1, \dots, v_5\}$ and that those vertices are $v_1, \dots, v_4$. This means in particular that there exists a shortest path $P$ covering $v_5$ in $D'$. Since $v_5$ and $v_6$ are twins in $D$, the path obtained by replacing $v_5$ by $v_6$ in $P$ is a valid path in $D$, and is of same length as $P$. It is thus a shortest path covering $v_6$ and $S$ is a geodetic set of $D$ of size (at most) $k$. This means that applying \Cref{rule-not-shortcutting} creates an equivalent instance each time it is used.

The previous arguments prove that the kernelization algorithm using \Cref{rule-shortcutting,rule-not-shortcutting} is correct. Now, let us compute the number of remaining vertices of the digraph $D$ after applying the aforementioned algorithm. Consider an instance $D'$ where no previous rule can be applied. If $k \leq 0$, we return a trivial \no-instance (for example a single vertex). If $k> 0$, let us denote by $X$ a minimum-size vertex cover of $D'$. Because none of the rules can be applied on $D'$, any  false twin set is of size at most five. By partitioning the independent set $V\setminus X$ into inclusion-wise maximal  false twin sets, we obtain:
\[
|V(D')| \leq |X| + 5\cdot 4^{|X|} = 2^{O(\vcn)}.\qedhere
\]
\end{proof}

The next proposition states the existence of the algorithm to be applied on the above kernel. It is the last element necessary for the proof of \Cref{thm-vc-algo}.

\begin{proposition}\label{prop-algo-exp-DGS}
\gsfull admits an algorithm running in time $n^{O(\vcn)}$.
\end{proposition}

\begin{proof}
Let $(D,k)$ be an instance of \gsfull. Compute in time $2^{O(\vcn)}\cdot n^{O(1)}$ a minimum-size vertex cover $X$ of $D$ using some \FPT\ algorithm (see \cite{harris_et_al2024}). Let $I=V(D) \setminus X$. Compute then in polynomial time the set $S$ of all vertices in $I$ whose neighborhood is shortcutting. Those vertices belong to any geodetic set of $D$. Now, notice that $X\cup S$ is a geodetic set of $D$: any vertex $v\in I\setminus S$ has two different neighbors $x\in N^-(v)$ and $y\in N^+(v)$ that both belong to $X$, so that the arc $xy$ does not belong to $D$. Indeed, the vertices in $I$ for which no such neighbors can be found are exactly the ones with a shortcutting neighborhood, \emph{i.e.} vertices of $S$. Thus $v$ is covered by the shortest path $(x, v, y)$. With this observation, we can enumerate all subsets $S'$ of $V(D) \setminus S$ with size at most $|X|$ and check whether the set $S\cup S'$ is a geodetic set in polynomial time. If at some point, one is a geodetic set of size at most $k$, answer \yes, and answer \no otherwise. The total running time of this algorithm is then $2^{O(\vcn)}\cdot n^{O(1)} + n^{O(1)} + n^{O(\vcn)} = n^{O(\vcn)}$.
\end{proof}

This implies the proof of the first part of \Cref{thm-vc-algo}, which we state below.

\begin{proof}[Proof of \Cref{thm-vc-algo} (algorithm)]
Let $(D,k)$ be an instance of \gsfull and denote by $\vcn$ the vertex cover number of $D$. First, we apply the kernelization algorithm presented before to obtain an equivalent instance $(D',k')$ of size $|V(D')| = 2^{O(\vcn)}$ in polynomial time. Then, we apply the algorithm presented in the proof of \Cref{prop-algo-exp-DGS} on instance $(D', k')$ and return the answer obtained. Correctness of the algorithm is ensured by \Cref{prop-kernel,prop-algo-exp-DGS}. This last step is executed in time $(2^{O(\vcn)})^{O(\vcn)} = 2^{O(\vcn^2)}$, thus ending the proof.
\end{proof}

\subsection{\texorpdfstring{\ETH}{ETH}-based lower bound}\label{sub-vcn}

In this subsection, we prove the second part of \Cref{thm-vc-algo}, \emph{i.e.} that the algorithm obtained in \Cref{prop-algo-exp-DGS} is optimal under \ETH \cite{ETH1999}. To this end, we present a reduction from the following variant of \textsc{3-SAT} defined in \cite{lampis2025}.

\problemdc{\partsat}{A formula $\varphi$ in 3-CNF form, together with a partition of the set of its variables $\mathcal X$ into three disjoint sets $X, Y, Z$, with $|X| = |Y| = |Z| = N$, such that no clause contains more than one variable from each $X, Y$ and $Z$.}{Determine whether $\varphi$ is satisfiable.}{0.95} 

\begin{proposition}[\cite{lampis2025}]\label{prop-ETH-3SAT}
Unless \ETH fails, \partsat does not admit an algorithm running in time $2^{o(N)}$.
\end{proposition}

\subparagraph*{Reduction description}

Consider an instance $\mathcal I = (X, Y, Z,\mathcal C)$ of 3-SAT where $\mathcal X = X\cup Y\cup Z$ is the variable set, with $|X| = |Y| = |Z| = N$. By adding dummy variables, if necessary, we can assume that \(N = n^2\) for some integer \(n\).
Let $\mathcal C = \{C_1, \dots, C_m\}$ be the set of clauses of $\mathcal I$. We create a digraph $\mathcal D=(V, A)$  through the following procedure. The construction is highly symmetric with respect to the partition of $\mathcal X$, so we describe it only for the set $X$, for the sake of clarity.

\begin{itemize}
    \item Partition the variables of $X$ into buckets $X^1, \dots, X^{n}$, such that each bucket contains $n$ variables. Without loss of generality, we assume an arbitrary but fixed ordering of buckets and variables inside each bucket.
    We use the notation \(X^i_j\) to denote $j$-th variable of the \(i\)-th bucket. Throughout the construction, we follow the notation that superscript will corresponds the index of the bucket while 
    subscript denotes the vertex index itself.
    
    \item For every bucket $X^i$, add to $V(\mathcal D)$ a set $A^i$ of $2^{ n}$ vertices corresponding to each possible assignment of the variables of the bucket. Denote those vertices by $a^i_{p}$ for $p\in [2^{n}]$.
    
    \item For each $i\in [n]$, add two new vertices $\alpha_{A ^i}$ and $\beta_{A^i}$. For each \(i \in [n]\), add the arcs $\alpha_{A^i}a^i_{p}$ and $a^i_{p}\beta_{A^i}$ for each $p\in [2^{n}]$.
    
    \item Add vertex sets $V$, $R$, $T$ and $F$ on $n$ vertices each. Denote by $v_i$ (respectively $r_i$, $t_i$, $f_i$) the vertices of $V$ (respectively $R$, $T$, $F$) for $i\in[n]$. Add the arcs $a^i_{p}v_i$ for each $i\in [n]$ and $p\in[2^{ n}]$. Add also the arcs $v_i r_{j}$ for all $i,j\in [n]$ such that $i\not=j$. Finally, add the arcs $\alpha_{A_i}v_j$ and $\alpha^i_{p}r_j$ for all $i,j\in [n]$ and $p\in[2^{ n}]$.
    
    \item For every \(i \in [n]\) and every \(p \in [2^n] \), connect $a^i_{p}$ with vertices in $T\cup F$ such that the connection represents the values of variables set $X^i$ by the assignment corresponding to $a^i_{p}$. Formally, the reduction adds the following arcs.
    \begin{itemize}
        \item If the partial assignment associated to $a^i_{p}$ is setting variable $X^i_j$ to \TRUE, then we add the arc $a^i_{p}t_j$.
        See the green arcs in Figure~\ref{fig-reduction}.
        \item If the partial assignment associated to $a^i_{p}$ is setting variable $X^i_j$ to \FALSE, then we add the arc $a_{i,p}f_j$.
        See the red arcs in Figure~\ref{fig-reduction}.
    \end{itemize}
    Note that $a^i_{p}$ is connected to exactly \(n\) vertices in \(T \cup F\). More precisely for every \(j \in [n]\), $a^i_{p}$ is connected to \(t_j\) if and only if it is \emph{not} connected to \(f_j\), and vice versa.
    \item Add vertex sets $U$ and $W$, both on $n$ vertices. Denote by $u_i$ (respectively $w_i$) the vertices of the set $U$ (respectively $W$) for $i\in [ n]$. Add the arcs $u_iw_i$ and $r_iw_i$ for all $i\in [n]$.

    \item For all sets $S$ in $\{T, F\}$, add two vertices $\alpha_{S}$ and $\beta_{S}$, as well as a vertex $\alpha_{U}$. For $S\in\{T,F,U\}$, add the arcs $\alpha_Ss_i$ and $s_i\beta_S$ (when $\beta_S$ exists) for all $i\in [n]$, where $s_i$ designates the vertex of $S$ with index $i$.
    
\end{itemize}

The previous construction is duplicated twice, to encode variables of $Y$ and $Z$ as well. When the naming of vertices of the previously described gadgets is ambiguous, we add a superscript to indicate the variable set it corresponds to. For instance, the vertex $v_i^Y$ corresponds to the vertex $v_i$ of the above construction for the variable set $Y$. For sets $A^i$ and vertices $a^i_p$, instead of using a single index as superscript, we give directly the bucket to which it is associated. For instance, vertex $a^{Z^3}_5$ corresponds to the fifth assignment of bucket $Z^3$. Additionally, call $B$ the set of all vertices $\alpha_S$ and $\beta_S$ mentioned during the construction. 

The end of the construction focuses on encoding the clauses as follows. Recall that $\mathcal I$ is an instance of \partsat, each clause of $\mathcal C$ contains exactly one variable from $X$, one from $Y$ and one from $Z$.
\begin{itemize}
    \item Consider a clause \(C_{\ell} \in \mathcal C\), and denote literals in \(C_\ell\) (when they exist) by $X^{i_x}_{j_x}$, $Y^{i_y}_{j_y}$ and $Z^{i_z}_{j_z}$ respectively.  
Add a vertex \(c_{\ell}\) in the graph and add the following arcs.
\begin{itemize}
    \item If assigning \TRUE (respectively \FALSE) to variable $X^{i_x}_{j_x}$ (when it exists) satisfies the clause $C_\ell$, then add the arc $t_{j_x}c_\ell$ (respectively $f_{j_x}c_\ell$) and the arc $c_\ell u_{i_x}$.
    \item When they exist, perform a similar operation with variables $Y^{i_y}_{j_y}$ and $Z^{i_z}_{j_z}$.
\end{itemize}
\end{itemize}

Let \(C\) be the set of all the vertices added with respect to some clause. Note
that for every vertex in \(C\), the number of its in-neighbors is equal to the number of its out-neighbors, which is equal to number of literals in it. We will argue that a vertex \(c_{\ell}\) is in a shortest path from a vertex  \(a^i_p\) , corresponding to a (partial) assignment, and some sink vertex  if and only if the (partial) assignment corresponding to \(a^i_p\) satisfies the clause. We conclude the construction by adding some arcs to prevent `unwanted coverage'.
\begin{itemize}

    \item Add for all $i\in [n]$ and all $S,S'\in \{X,Y,Z\}$ the arcs $\alpha_{A^{S^i}}v_i^{S'}$, $\alpha_{T^S}u_i^{S'}$ and $\alpha_{F^S}u_i^{S'}$. These arcs will prevent any coverage from a source vertex associated with the variable set $S$ to cover clauses with a path going to a sink associated with variable set $S'$.
    \item Add for each $S,S'\in \{X, Y, Z\}$ such that $S \not = S'$ the arcs  $a^{S^i}_{p}u^{S'}_{j}$ for all $i, j\in [n]$ and all $p\in [2^n]$. These arcs will prevent any coverage from a selected vertex associated with the variable set $S$ to cover clauses with a path going to a sink associated with variable set $S'$.

\end{itemize}
This ends the construction of $\mathcal D$. The instance returned by this reduction is $(\mathcal D,k)$ with $k=12 n + 15$. \Cref{fig-reduction} is a partial representation of the obtained instance.

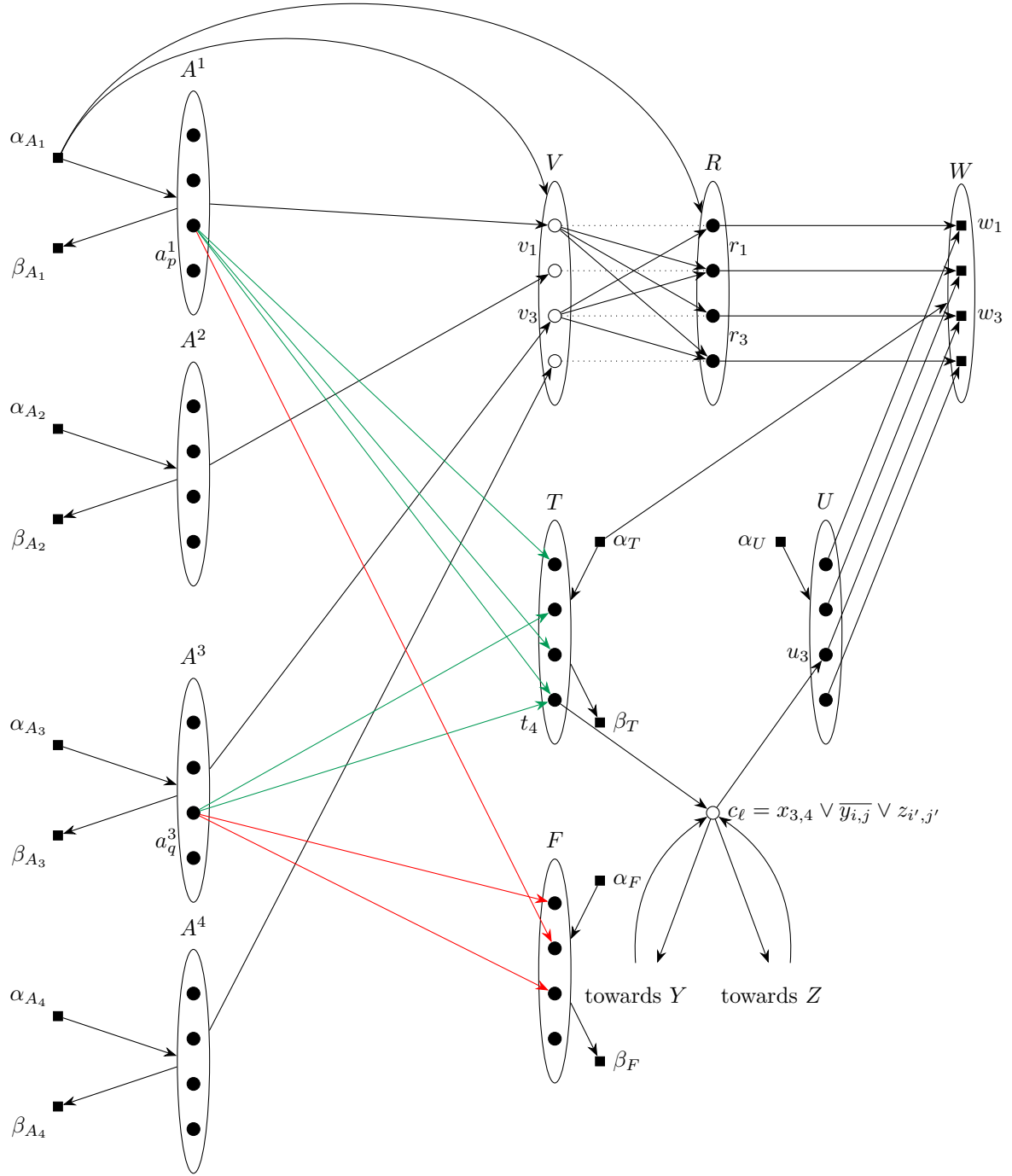
\begin{figure}
    \begin{adjustbox}{scale=1, center}
        \input{images/reduc-vc-MFCS}
    \end{adjustbox}
    \caption{Partial representation of the instance of \gsfull{} obtained from \partsat{} by the reduction described in \Cref{sub-vcn} for $n=4$. 
    If an arc goes from one vertex $v$ to a vertex set $S$, it represents all arcs $vv'$ with $v'\in S$ (similarly for an arc from one vertex set to a set).
    Sources and sinks are represented with squares, while empty circles represent vertices that are \emph{not} covered by shortest paths from sources to sinks. 
For the sake of clarity, we only show representative arcs. For example, only vertices associated to variables of $X$ and a single clause vertex $c_\ell$ are represented.
  The arc from \(\alpha_{A_1}\) represent the arcs of the form \(\alpha_{A_i}\) to \(V\) and \(\alpha_{A_i}\) to \(R\).
   The arcs added while encoding clauses from vertices associated to variables $X$ to vertices associated to other variable sets (namely $Y$ and $Z$) are \emph{not} represented. Arcs between $V$ and $R$ are represented only for $v_1$ and $v_3$ as starting vertices and dotted arcs represent missing arcs.} 
    \label{fig-reduction}
\end{figure}

\subparagraph*{Correctness of the reduction}

In Lemma~\ref{lemma-SAT-to-DGS}, we prove 
that if $\mathcal I$ is satisfiable, then $\mathcal D$ admits a geodetic set of size $k$.
In this section, we present an outline of the
proof in the reverse direction.

We first present an informal overview. It is easy to see that all sources and sinks must belong to any geodetic set of $\mathcal D$ (\Cref{lem-sink-sources}). We argue that the shortest paths between the sinks and sources cover all the vertices in $V(\mathcal D)$, except the vertices in $C$ and $\bigcup_{S\in\{X,Y,Z\}}V^S$. For example, consider vertex $c_{\ell}$ and set $V$ in Figure~\ref{fig-reduction}. We use the simplified notation from the figure to complete the overview.

To cover vertices $v_1, v_2, v_3, v_4$, and vertex $c_{\ell}$, any geodetic set must include \emph{at least one} vertex from the sets $A^1, A^2, A^3$, and $A^4$, respectively. The cardinality constraint on the geodetic set ensures that it contains \emph{exactly} one vertex from each set. This selection naturally corresponds to a valid assignment of the variables in the instance $\mathcal I$. We will argue that this must indeed be a satisfying assignment.

For vertices corresponding to assignments (such as $a^1_p$ and $a^3_q$ in Figure~\ref{fig-reduction}), their variable assignments are evident from their adjacencies to vertices in $T$ and $F$. 
These two sets relay the assignment information to the clause vertices. 
Note that both $a^1_p$ and $a^3_q$ set the fourth variable in their respective buckets to \TRUE. 
Consider the clause $c_{\ell}$, which contains the variable $X^3_4$. To obtain a valid assignment, we must ensure that the shortest path between $a^3_q$ and some sink vertex covers the vertex $c_{\ell}$. Moreover, the shortest path between $a^1_p$ and the sink vertices must \emph{not} cover $c_{\ell}$. 
The relevant sink vertices correspond to the set $W$.

These conditions are enforced by the vertices in sets $V$ and $R$, along with their interconnecting arcs. Consider the two paths $a^1_p-t_4-c_{\ell}-u_3-w_3$ and $a^3_q-t_4-c_{\ell}-u_3-w_3$.
Note that the first path is \emph{not} a shortest path because there is a shorter alternative: $a^1_p-v_1-r_3-w_3$. However, there is no analogous shortcut path $a^3_q-v_3-r_3-w_3$ because the arc from $v_3$ to $r_3$ does not exist. This distinction forms the crux of the proof of correctness.

In the remainder of this section, we formalize this intuition. 
Before proceeding, we remark that this missing arc 
and the reliance on these connections to bypass 
$V$ for shortest paths critically depends on the fact that clause $c_{\ell}$ contains at most one literal from the set $X$.

In the considered construction, we define set 
\(\mathcal S_0\) as the collection of sinks
and sources, i.e.,  $S_0 := B\cup \bigcup_{S\in \{X,Y,Z\}}W^S$. 
In the next two observations, we argue that vertices of $\mathcal S^0$ cannot cover vertices in $C$ and that vertices of $\bigcup_{S\in\{X,Y,Z\}}V^S$ can only be covered by paths starting at specific vertices not belonging to $\mathcal S^0$.

\begin{lemma}\label{lemma-SAT-to-DGS}
If $\mathcal I$ is satisfiable, then $\mathcal D$ admits a geodetic set of size $k$.
\end{lemma}

\begin{proof}
Suppose $\varphi : X \to \{\text{\TRUE, \FALSE}\}$ is a satisfying assignment for $\mathcal I$. We construct a geodetic set $\mathcal S$ of size $k = 12n + 15$ based on this assignment. First define the set  $\mathcal S^0 = B\cup \left(\bigcup_{s\in\{X,Y,Z\}} W^S\right)$. With this, $|\mathcal S^0| = 9n+15$, and we will add $3 n$ new vertices to obtain our set $\mathcal S$. Consider a bucket of variables $X_i$ for some $i\in [n]$ and denote the partial assignment of its variables in $\varphi$ by $\varphi_i^X$. By construction of $\mathcal D$, there exists in $A_i^X$ a vertex $a_{i,p}^X$ corresponding to this assignment $\varphi_i^X$ with $p\in [2^{ n}]$. We add all such vertices to $\mathcal S^0$, and obtain the set $\mathcal S$ such that  $|\mathcal S| = 12n+15$.

Next, we prove that $\mathcal S$ defined above is a valid geodetic set of $\mathcal D$. Let us describe some shortest paths and check that each vertex of $\mathcal D$ is covered (we do not check vertices of $\mathcal S^0$ since they belong to $\mathcal S$). In the following, we only check vertices associated to variables in $X$, as the proof is similar for variables in $Y$ and $Z$. If a single vertex or a vertex set is named without any superscript set, we consider it is associated with $X$.

\begin{itemize}
    \item The shortest paths $\left\{(\alpha_{A^i},  a^i_p, \beta_{A^i}) \mid i\in [n],\;p\in[2^n]\right\}$ cover all vertices of sets $A^i$ for all $i\in[n]$.
    \item The shortest paths $\left\{(\alpha_{T}, t_i, \beta_{T}) \mid i\in[n]\right\}$ cover all vertices of $T$ . The same goes for paths $\left\{(\alpha_{F}, f_i, \beta_{F}) \mid i\in[n]\right\}$ that cover vertices of $F$.
    \item Recall that for each $i\in[n]$, there exists a vertex $a^i_{p}\in S\cap A^i$ by construction of $\mathcal S$. Then, the shortest paths $\left\{(a^i_{p}, v_i, r_{j}, w_{j}) \mid i,j\in[n],\; i\not = j\right\}$ of size $3$ cover all vertices of $V$ and $R$.
    \item The paths $\{(\alpha_U,u_i, w_i)\mid i\in [n]\}$ of size 2 cover all vertices of $U$. 
    \item Consider a clause $C_\ell\in \mathcal C$ and its associated vertex $c_\ell$ for some $\ell\in[m]$. Because $\varphi$ is satisfying $\mathcal I$, there exists a variable $\chi\in \mathcal X$ that satisfies $C_\ell$ under the assignment $\varphi$. Without loss of generality, say that $\chi = X^i_{j}\in X^i$ for some $j\in[n]$. Because of this observation, the vertex $a^i_{p}$ of $\mathcal S$ belonging to $A^i$ is associated with a partial assignment $\varphi^i$ 
    that sets the variable $\chi$ in a way that satisfies $C_\ell$. 
    This means that there exists a vertex $\gamma_j\in T\cup F$ so that the arcs $a^i_{p}\gamma_j$,  $\gamma_jc_\ell$  and $c_\ell u_i$ are present in $\mathcal D$. Then, the path $(a^i_{p}, \gamma_j, c_\ell, u_i, w_i)$ is a valid path of size 4 in $\mathcal D$. 
    Furthermore, notice that by construction of $\mathcal D$, no path from $a^i_{p}$ can reach $w_i$ by going through the set $V$. This means that the path described previously is a shortest path in $\mathcal D$, and vertices of $C$ are all covered. 
\end{itemize}
Since all vertices of $\mathcal D$ are covered, $S$ is a geodetic set of $\mathcal D$ of size $k = 12 n+15$.
\end{proof}

\begin{observation}\label{obs-C-not-covered}
Let $S\in\{X,Y,Z\}$. No shortest path between vertices of $\mathcal S^0$ covers vertices of $C$.
\end{observation}

\begin{proof}
    Observe that all vertices of $\mathcal S^0$ are either sources or sinks in $\mathcal D$, and the only sources are vertices $\alpha_S$  where $S \in \{A^{\Sigma^i}, T^{\Sigma}, F^{\Sigma}, U^{\Sigma}\mid \Sigma\in \{X,Y,Z\},\; i\in [n]\}$. Similarly, sinks are vertices of $W^{S}$ for all $S\in \{X,Y,Z\}$ and vertices $\beta_S$ where $S \in \{A^{\Sigma^i}, T^{\Sigma}, F^{\Sigma}\mid \Sigma\in \{X,Y,Z\},\; i\in[n]\}$. Vertices $\alpha_{U^{\Sigma}}$ cannot reach any vertex of $C$, so they cannot give rise to such a covering shortest path for vertices of $C$. Consider now some $\alpha_{T^{\Sigma}}$. The vertices that can be reached from it are exactly the vertex $\beta_{T^{\Sigma}}$ through a path of size two, and vertices $u_i^{\Sigma'}$ for all $i\in [n]$ and $\Sigma'\in \{X,Y,Z\}$ through a path of size one. Since vertices of $C$ are not adjacent to $\alpha_{T^{\Sigma}}$, no vertex of $C$ can be covered by such described paths. The same applies to vertices $\alpha_{F^{\Sigma}}$ through similar arguments. Consider now $\alpha_{A^{\Sigma^i}}$, and notice that no path from it to a vertex $\beta_{S}$ with $S \in \{A^{\Sigma^i}, T^{\Sigma}, F^{\Sigma}\mid \Sigma\in \{X,Y,Z\}\}$ crosses a vertex of $C$. Finally, since for all $\Sigma,\Sigma'\in \{X,Y,Z\}$ and all $i\in [n]$, the arc $\alpha_{A^{\Sigma^i}}v_i^{\Sigma'}$ is present in $\mathcal D$, we have $d(\alpha_{A^{\Sigma^i}}, u_i^{\Sigma'}) = 2$ and by construction, $d(\alpha_{A^{\Sigma^i}}, c_\ell) \geq 3$ for all $\ell\in [m]$. This means no path between vertices of $\mathcal S^0$ can cover a vertex of  $C$.
\end{proof}

\begin{observation}\label{prop-cover-of-V}
Let $S\in\{X,Y,Z\}$. The vertex $v_i^S$ can only be covered by a path that starts from a vertex of $A^{S^i}\cup\{v_i^S\}$.
\end{observation}

\begin{proof}
Consider the vertex $v_i^S$ for some $S\in \{X,Y,Z\}$ and some $i\in [n]$.  The only vertices that are able to reach $v_i^S$ in $\mathcal D$ are $\alpha_{A^{S^i}}$, $a_{p}^{S^i}$ for all $p\in [2^n]$ and $v_i^S$ itself. We show that no path starting from $\alpha_{A^{S^i}}$ can cover $v_i^S$. Indeed, in order to cover $v_i^S$, the path has to end in some vertex reachable from $v_i^S$, namely either a vertex of $R^S$ or $W^S$, except vertices $r_i^S$ and $w_i^S$ that are unreachable. Since the arc $\alpha_{A^{S^i}}r_{j}^S$ is present in $\mathcal D$ for all $j\in [n]$, neither the shortest paths from $\alpha_{A^{S^i}}$ to some $r_{j}^S$ nor the ones from $\alpha_{A^{S^i}}$ to $w_{j}^S$ are going through $v_i^S$ (the first ones are of length 1 and the latter are of length 2, while $d(\alpha_{A^{S^i}}, v_i^S) = 2$). 
\end{proof}

We also argue that $\mathcal D$ is an acyclic digraph by using another simple observation.

\begin{observation}\label{obs-layers}
    Digraph $\mathcal D$ can be decomposed into layers $L_1, \dots, L_5$ such that all arcs of $A(\mathcal D)$ are either internal to a layer and adjacent to a leaf vertex, or going from layer $L_i$ to layer $L_j$ with $i < j$. In particular, $\mathcal D$ is acyclic.
\end{observation}

\begin{proof}
    Consider the following layers:
    \begin{description}
        \item[$L_1$] is composed of vertices $\alpha_{A^{S^i}}, \beta_{A^{S^i}}$ and vertices in sets $A^{S^i}$ for all $i\in [n]$ and $S\in \{X,Y,Z\}$.
        \item[$L_2$] is composed of vertices in sets $V^S$, $T^S$, $F^S$ and vertices $\alpha_{T^S}$, $\alpha_{F^S}$,$\beta_{T^S}$, $\beta_{F^S}$ for all $S\in \{X,Y,Z\}$.
        \item[$L_3$] is composed of of vertices in sets $R^S$ and $C$ for all $S\in \{X,Y,Z\}$.
        \item[$L_4$] is composed of vertices in $U^S$ and vertices $\alpha_{U^S}$ and $\beta_{U^S}$ for all $S\in \{X,Y,Z\}$.
        \item[$L_5$] is composed of vertices in $W^S$ for all $S\in \{X,Y,Z\}$.     
    \end{description}

    All vertices of $\mathcal D$ are belonging to a unique layer. Furthermore, it can be verified clearly that any arc in $A(\mathcal D)$ internal to one of the layers is adjacent to either vertex $\alpha_S$ of a vertex $\beta_S$ for some vertex set $S$, which are leaf vertices in $\mathcal D$. Finally, by construction, no arc goes from a layer $L_i$ to a layer $L_j$ with $i > j$. This proves that we can define a topological ordering on $V(\mathcal D)$ by ordering first vertices $\alpha_S$, ordering last vertices  $\beta_S$, and considering in between the vertices based on the index of their layer.
    
\end{proof}

\begin{lemma}\label{lemma-DGS-to-SAT}
If $\mathcal D$ admits a geodetic set of size $k$, then $\mathcal I$ is satisfiable.
\end{lemma}

\begin{proof}
Suppose that $\mathcal D$ admits a geodetic set $\mathcal S$ of size $k=12n+15$. 
By \Cref{lem-sink-sources}, all sources and sinks of $\mathcal D$ are necessarily in $\mathcal S$. The previous remark is equivalent to $\mathcal S^0\subset \mathcal S$. By \Cref{prop-cover-of-V}, there is at least one vertex from $A^{S^i}\cup\{v_i^S\}$ in $\mathcal S$ for all $i\in [n]$ and all $S\in \{X,Y,Z\}$. Since $|\mathcal S^0| = 9n+15$ and  $k = 12n+15$, $\mathcal S$ contains exactly one vertex from each of the sets $A^{S^i}\cup \{v_i^S\}$. By \Cref{obs-C-not-covered}, and since $\mathcal S$ is a geodetic set, vertices of $C$ must be covered by a path starting in $\mathcal S \setminus \mathcal S^0$. Furthermore, no vertex of $C$ is reachable from a vertex of $V^S$ in $\mathcal D$, so $\mathcal S$ contains exactly one vertex from each vertex set $A^{S^i}$ for $i\in [n],\; S\in \{X,Y,Z\}$.

Consider for some $\ell\in[m]$ the vertex $c_\ell$. This vertex is covered by some vertex  in set $$\hat  {\mathcal S}=  \mathcal S\cap\bigcup_{\substack{i\in [n]\\S\in\{X,Y,Z\}} }A^{S^i}.$$
Without loss of generality, say it is $a_{p}^{X^i}\in A^{X^i}$ for some $i\in [n]$ and some $p\in[2^{ n}]$. The remainder of the proof can be adapted easily to variable sets $Y$ and $Z$. By construction of $\mathcal D$, this means that there exists a vertex $\gamma_j^X\in T^X\cup F^X$ and a vertex $u_i^X\in U^X$ so that arcs $a_{p}^{X^i}\gamma_j^X$, $\gamma_j^Xc_\ell$ and $c_\ell u_i^X$ are present in $\mathcal D$. Because of the construction of $\mathcal D$, this means that the partial assignment of the variable $X^i_{j}$ corresponding to vertex $a_{p}^{X^i}$ satisfies the clause corresponding to $c_\ell$. Since this observation can be done for each vertex of $C$, we can define the assignment $\varphi : X \to \{\text{\TRUE, \FALSE}\}$ from the partial assignments associated with the vertices of $\hat {\mathcal S}$. This assignment is well defined because there is exactly one vertex of each $A^{S^i}$ in $\mathcal S$. As stated before, each vertex of $C$ is covered, thus the assignment $\varphi$ satisfies all the clauses in $\mathcal C$. 
\end{proof}

\begin{proof}[Proof of \Cref{thm-vc-algo} (optimality)]
The reduction presented above takes an instance $\mathcal I$ of \partsat and by \Cref{lemma-DGS-to-SAT,lemma-SAT-to-DGS} returns an equivalent instance $(\mathcal D,k)$ of \gsfull in time $2^{O(n)}$ with $|V(\mathcal D)| = 2^{O( n)}$. Furthermore, by \Cref{obs-layers}, $\mathcal D$ is a DAG. Note that taking all the vertices in layers $L_2, L_4$ and $L_5$ from the proof of \Cref{obs-layers} as well as vertices of $B$ results in a vertex cover of the underlying graph of $\mathcal D$. This means that 

$$\vcn(\mathcal D) \leq |B| + \sum_{S\in\{X,Y,Z\}}\left(|V^S| + |T^S| + |F^S| + |U^S| + |W^S| \right)= 21n+15.$$

We thus get that $\vcn(\mathcal D) + k = O(n)$. Now, suppose for a contradiction that ETH holds and that there exists an algorithm to solve \gsfull in time $2^{o(\vcn^2)}\cdot n^{O(1)}$. Consider the following algorithm to solve \partsat on the instance $\mathcal I$ on $3N$ variables. First, use the reduction described above to construct an equivalent instance $(\mathcal D,k)$ of \gsfull with $\vcn(\mathcal D) + k = O(n)$ and $|V(\mathcal D)| =2^{O(n)} = 2^{o(n^2)} = 2^{o(N)}$. Then, solve \gsfull on this new instance with the assumed algorithm in time

$$2^{o(\vcn^2)}\cdot |V(\mathcal D)|^{O(1)} = 2^{o(N)}\cdot 2^{o(N)} = 2^{o(N)}$$
and return the same answer for instance $\mathcal I$. This algorithm is correct by our  assumption and by \Cref{lemma-SAT-to-DGS,lemma-DGS-to-SAT}, and its total running time is $2^{o(N)}$, thus it contradicts \Cref{prop-ETH-3SAT}.
\end{proof}

%% file: images/reduc-vc-MFCS.tex
\newcounter{iSet}
\newcounter{leftNode}

\NewDocumentCommand{\drawISet}{m m m O{}}{%
    \stepcounter{iSet}
    \begin{scope}[mynode/.append style={ circle}, every node/.style={mynode}, #4]
        \node (n\theiSet_1) at ($(0,1.5) + (#1,#2)$) {};
        \node (n\theiSet_2) at ($(0,.5) + (#1,#2)$) {};
        \node (n\theiSet_3) at ($(0,-.5) + (#1,#2)$) {};
        \node (n\theiSet_4) at ($(0,-1.5) + (#1,#2)$) {};
        
        \node[fill=none, draw, shape=ellipse, fit=(n\theiSet_1) (n\theiSet_2) (n\theiSet_3) (n\theiSet_4)] (s\theiSet) [label=#3]{};
    \end{scope}}
    
\NewDocumentCommand{\drawLeftNodes}{m m m O{}}{%
    \stepcounter{leftNode}
    \begin{scope}[mynode/.append style={circle,fill}, every node/.style={mynode}, #4]
        \begin{scope}[mynode/.append style={rectangle}]
            \node (alpha#3) [label=above left:$\alpha_{A_{#3}}$] at ($(-3,1) + (#1,#2)$) {};
            \node (beta#3) [label=below left:$\beta_{A_{#3}}$] at ($(-3,-1) + (#1,#2)$) {};
        \end{scope}
    \end{scope}}
    
\begin{tikzpicture}[scale = .70,
    every edge/.append style = {draw}, >={Stealth[scale=1.2]},
    mynode/.style={inner sep=2pt,draw=black}]
    \begin{scope}[mynode/.append style={fill}]
    \drawISet{-4}{9.5}{$A^1$}
    \drawISet{-4}{3.5}{$A^2$}
    \drawISet{-4}{-3.5}{$A^3$}
    \drawISet{-4}{-9.5}{$A^4$}
    \end{scope}
    
    \drawISet{4}{7.5}{$V$}
    \begin{scope}[mynode/.append style={fill}]
    \drawISet{4}{0}{$T$}
    \node[mynode, label=right:$\alpha_T$] (alpha_t) at (5,2) {};
    \node[mynode, label=right:$\beta_T$] (beta_t) at (5,-2) {};
    \drawISet{4}{-7.5}{$F$}
    \node[mynode, label=right:$\alpha_F$] (alpha_f) at (5,-5.5) {};
    \node[mynode, label=right:$\beta_F$] (beta_f) at (5,-9.5) {};
    \drawISet{7.5}{7.5}{$R$}
    \drawISet{10}{0}{$U$}
    \node[mynode, label=left:$\alpha_U$] (alpha_u) at (9,2) {};
    \drawISet{13}{7.5}{$W$}[every node/.append style={rectangle}]
    \end{scope}

    \node[mynode/.append style={circle}, mynode, label={right:$c_\ell = x_{3,4}\vee \overline{y_{i,j}}\vee z_{i',j'}$}] (c) at (7.5,-4) {};

    \node[label=below left:$a^1_{p}$] at (n1_3) {};
    \node[label=below left:$a^3_{q}$] at (n3_3) {};

    \node[label=below left:$v_{1}$] at (n5_1) {};
    \node[label=left:$v_{3}$] at (n5_3) {};

    \node[label=below right:$r_{1}$] at (n8_1) {};
    \node[label=below right:$r_{3}$] at (n8_3) {};

    \node[label=right:$w_{1}$] at (n10_1) {};
    \node[label=right:$w_{3}$] at (n10_3) {};

    \node[label=below left:$t_{4}$] at (n6_4) {};
    \node[label=left:$u_{3}$] at (n9_3) {};
    

    \node[label={below:towards $Y$}] (y1) at (5.8,-7.5) {};
    \node[] (y2) at (6.2,-7.5) {};
    \node[label={below:towards $Z$}] (z1) at (8.8,-7.5) {};
    \node[] (z2) at (9.2,-7.5) {};

    \path[->] (c) edge[] (y2);
    \path[->] (y1) edge[bend left] (c);
    \path[->] (c) edge[] (z1);
    \path[->] (z2) edge[bend right] (c);

    \foreach \x/\y in {9.5/1, 3.5/2, -3.5/3, -9.5/4} {\drawLeftNodes{-4}{\x}{\y}[rectangle, fill]}

    \foreach \i in {1} {
        \path[->] (alpha\i) edge[bend left=70] (s5);
        \path[->] (alpha\i) edge[bend left=70] (s8);
    }

    \foreach \i in {1, ..., 4} {
        \path[] (n5_\i) edge[dotted] (n8_\i);
        \path[->] (n8_\i) edge (n10_\i);
        \path[->] (n9_\i) edge (n10_\i);
        \path[->] (s\i) edge (n5_\i);
    }

    \foreach \i in {1, 2, 4} {
        \path[->] (n5_3) edge (n8_\i);
    }

    \foreach \i in {2, 3, 4} {
        \path[->] (n5_1) edge (n8_\i);
    }

    \foreach \i/\j in {t/6,f/7} {
       \path[->] (alpha_\i) edge (s\j);
       \path[->] (s\j) edge (beta_\i);
    }
    \path[->] (alpha_u) edge (s9);
    \path[->] (alpha_t) edge (s10);

    \foreach \x in {1,...,4} {\path[->] (alpha\x) edge (s\x);}
    \foreach \x in {1,...,4} {\path[->] (s\x) edge (beta\x);}

    \foreach \x in {2,4} {\path[->, ForestGreen] (n3_3) edge (n6_\x);}
    \foreach \x in {1,3} {\path[->, red] (n3_3) edge (n7_\x);}

    \foreach \x in {1,3,4} {\path[->, ForestGreen] (n1_3) edge (n6_\x);}
    \foreach \x in {2} {\path[->, red] (n1_3) edge (n7_\x);}

    \path[->] (n6_4) edge (c);
    \path[->] (c) edge (n9_3);
 
\end{tikzpicture}

%% file: kernel-algo.tex
\subsection{A kernel and an algorithm}

We now prove the first part of Theorem~\ref{thm-kernel}. Let $D$ be an input digraph to \gsfull with maximum degree $\Delta$ and reachability diameter $\mfd$. Consider the partition of $D$ into strongly connected components $C_1,\ldots,C_t$. We define the digraph $H$ on vertex set $C_1,\ldots,C_t$, with an arc from $C_i$ to $C_j$ ($i\neq j$) if there exists an arc from a vertex in $C_i$ to a vertex in $C_j$ in $D$. Note that $H$ is a DAG, for otherwise some strongly connected component from the partition would not be maximal. Assume without loss of generality that $C_1,\ldots,C_p$ ($1\leq p\leq t$) are the sources of $H$. 

Let $v$ be any vertex of $D$, with $v\in C_i$ fir some $i\in [t]$. Then, there exists $C_j$ with $1\leq j\leq p$ (that is, $C_j$ is a source in $H$), such that $C_i$ is reachable from $C_j$ in $H$. For any $i$ with $1\leq i\leq p$, let $x_i$ be an arbitrary vertex of the source component $C_i$ and let $X=\{x_i \mid 1\leq i\leq p\}$. Then, every vertex in $D$ is reachable from at least one vertex in $X$.

Note that the maximum number of vertices reachable from any vertex of $D$ is at most $\sum_{i=0}^{\mfd}\Delta^i\leq \Delta^{\mfd+1}$. In any \yes-instance of \gsfull, we must have $p\leq k$, as any geodetic set must include some vertex in each strongly connected component that is a source of $H$. Hence, in a \yes-instance, the number of vertices must be at most $k\Delta^{\mfd+1}$.

The kernelization algorithm is hence straightforward: if $D$ has more than $k\Delta^{\mfd+1}$ vertices, we return a trivial \no-instance (such as the digraph with $k+1$ isolated vertices, $k$ being unchanged); otherwise, we return $(D,k)$ itself.

The algorithm is a simple brute-force algorithm over such a kernel, testing all the possible subsets of size at most $k$ as a potential geodetic set, and checking in polynomial time whether each of them is indeed a geodetic set or not. This runs in time at most $\binom{k\Delta^{\mfd+1}}{k}(k \Delta^{\mfd+1})^{O(1)}$ which is in $(k\Delta)^{O(\mfd\cdot k)}$.

%% file: fpt-lower-bound.tex
\subsection{\texorpdfstring{\ETH}{ETH}-based lower bound}
\label{subsec:eth-lower-bound}

In this subsection, we prove the second part of
Theorem~\ref{thm-kernel}. In particular, we show that the simple
algorithm obtained by applying brute force to the reduced instance
is optimal under the \ETH. 
To this end, we present a reduction from
\(\ell \times \ell\)-\textsc{Independent Set}, which we define formally
below and state the corresponding conditional lower bound.

\problemdc{\(\ell \times \ell\)-\textsc{Independent Set}}{A graph \(H\) together
with a partition of its vertex set
\(\langle V_1, V_2, \dots, V_\ell\rangle\), where each \(V_i\) is an
independent set of cardinality \(\ell\), for all \(i \in [\ell]\).}{
Determine whether there exists an independent set \(X \subseteq V(H)\)
such that \(|X \cap V_i| = 1\) for every \(i \in [\ell]\).}{1}

\begin{proposition}[Theorem 14.12 in \cite{DBLP:books/sp/CyganFKLMPPS15}]
\label{prop-eth-ind-set}
Unless the \ETH\ fails, \(\ell \times \ell\)-\textsc{Independent Set}
does not admit an algorithm running in time \(\ell^{o(\ell)}\).
\end{proposition}

The reduction takes as input an instance
\((H,\langle V_1, V_2, \dots, V_\ell\rangle)\)
of \(\ell \times \ell\)-\textsc{Independent Set},
runs in polynomial time, and outputs an instance
\((D,k)\) of \textsc{Geodetic Set} such that
\(k = 3\ell + 1\), the maximum degree of \(D\) satisfies
\(\Delta \le \ell^4\), and \(\mfd = 2\).

\begin{figure}[t]
\centering
\includegraphics[scale=0.75]{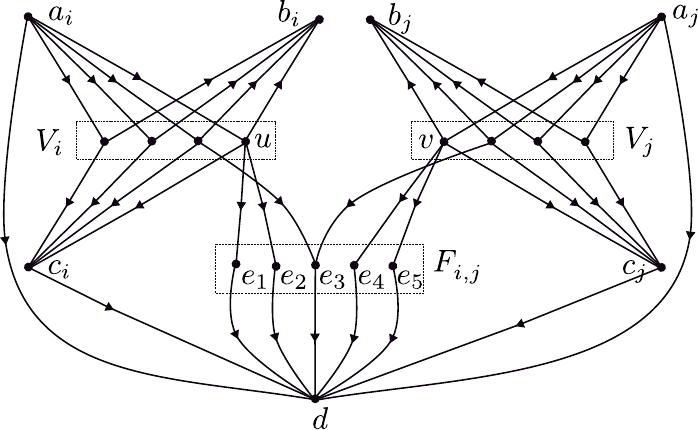}
\caption{Overview of the reduction from \(\ell \times \ell\)-\textsc{Independent Set} 
to \textsc{Geodetic Set}. 
Edges \(e_3, e_4, e_5\) are incident on \(u\), and 
edges \(e_1, e_2, e_3\) are incident on \(v\).
Note that for the sake of clarity, we do not show all the arcs.
\label{fig:eth-ind-set}}
\end{figure}

\begin{itemize}
  \item Initially, set
  \(V(D) := V(H)\) and \(A(D) := \emptyset\).

  \item For each part \(V_i\), introduce three new vertices
  \(a_i\), \(b_i\), and \(c_i\), and add the arcs
  \(a_iu\), \(ub_i\), and \(uc_i\) for every \(u \in V_i\).

  \item For every pair \(i \neq j \in [\ell]\), construct a set
  \(F_{i,j}\) as follows. For every edge \(uv \in E(H)\) with
  \(u \in V_i\) and \(v \in V_j\), add a corresponding
  vertex \(f_{uv}\) to \(F_{i,j}\).

  \item For every \(i \in [\ell]\) and every vertex \(w \in V_i\),
  add an arc \(wf_{uv}\) for every vertex
  \(f_{uv} \in F_{i,j}\) with \(j \in [\ell] \setminus \{i\}\),
  provided that \(w \notin \{u,v\}\).
  Equivalently, the arc \(wf_{uv}\) is added if and only if
  the edge \(uv\) is not incident to \(w\) in \(H\).

  \item Introduce a special vertex \(d\).
  For every \(i \neq j \in [\ell]\), add the arc \(f_{uv}d\)
  for every vertex \(f_{uv} \in F_{i,j}\).
  Moreover, for every \(i \in [\ell]\), add the arcs
  \(a_id\) and \(c_id\).
\end{itemize}

This completes the construction of the digraph \(D\).
The reduction outputs the instance \((D, k)\) with \(k = 3\ell + 1\)
of \textsc{Geodetic Set}.
In the following lemma, we prove the correctness of the reduction.

\begin{lemma}
\label{lemma-ind-set-geodetic-eth}
\((H,\langle V_1, V_2, \dots, V_\ell\rangle)\) is a \yes-instance of
\(\ell \times \ell\)-\textsc{Independent Set}
if and only if \((D,k)\) is a \yes-instance of \textsc{Geodetic Set}.
\end{lemma}

\begin{proof}
(\(\Rightarrow\))
Let \(X=\{u_1,\dots,u_\ell\}\), where \(u_i\in V_i\) for each \(i\in[\ell]\),
be an independent set in \(H\).
Consider the set
\(S := X \cup \{a_1,\dots,a_\ell\} \cup \{b_1,\dots,b_\ell\} \cup \{d\}\).
Clearly, \(|S| = 3\ell + 1\).
We show that \(S\) is a geodetic set of \(D\).

Every vertex \(u_i \in V(H) \subseteq V(D)\) lies on a shortest path
from \(a_i\) to \(b_i\) by construction.
Consider any edge \(vw \in E(H)\).
Since \(X\) is an independent set, there exists a vertex
\(u_i \in X\) that is not incident to the edge \(vw\).
By the construction of \(D\), there is a shortest path
\((u_i, f_{vw}, d)\), which covers the vertex \(f_{vw}\).
Finally, for every \(i \in [\ell]\), the vertex \(c_i\) lies on a shortest
path from \(u_i\) to \(d\).
Thus, every vertex of \(D\) lies on a shortest path between some pair
of vertices in \(S\), and hence \(S\) is a geodetic set of \(D\).

(\(\Leftarrow\))
Suppose that \(S\) is a geodetic set of \(D\) of size \(3\ell + 1\).
It must contain the set
\(S_0 := \{a_1,\dots,a_\ell\} \cup \{b_1,\dots,b_\ell\} \cup \{d\},
\)
since these vertices are sources or sinks in the digraph \(D\)
and must therefore belong to any geodetic set.

As shown in the forward direction, the vertices in \(S_0\) together
cover all vertices in \(V_i\) for every \(i \in [\ell]\).
Hence, the remaining vertices of \(S\) are used to cover vertices
in the sets \(\{c_i \mid i \in [\ell]\}\) and in the sets
\(F_{i,j}\) for \(i \neq j \in [\ell]\).

By construction, no shortest path between any pair of vertices in
\(S_0\) covers the vertex \(c_i\) for any \(i \in [\ell]\).
Therefore, \(S\) must contain at least one vertex from
\(\{c_i\} \cup V_i\) for each \(i \in [\ell]\).
Suppose that for some \(i \in [\ell]\), the vertex \(c_i \in S\).
Let \(u_i\) be any vertex in \(V_i\), and define
\(S' := (S \cup \{u_i\}) \setminus \{c_i\}\).
Observe that no shortest path starting or ending at \(c_i\)
covers any vertex in \(F_{i,j}\) for \(j \neq i\).
Moreover, the shortest path \((u_i, c_i, d)\) covers \(c_i\),
and thus replacing \(c_i\) by \(u_i\) preserves the geodetic property.
Hence, \(S'\) is also a geodetic set of \(D\).
Applying this argument for every \(i \in [\ell]\), we obtain a geodetic
set \(S'\) such that \(S' \cap V_i \neq \emptyset\) for all \(i \in [\ell]\).
Since \(S' \setminus S_0\) contains exactly \(\ell\) vertices,
it follows that \(S' \cap V_i\) contains exactly one vertex for each
\(i \in [\ell]\).

Let \(X\) be the set of vertices in \(S'\) that belong to
\(\bigcup_{i=1}^\ell V_i\).
We claim that \(X\) is an independent set in \(H\).
Suppose, for the sake of contradiction, that \(X\) is not independent.
Then there exists an edge \(uv \in E(H)\) such that
\(u \in V_i\) and \(v \in V_j\) for some distinct \(i,j \in [\ell]\).
Consider the vertex \(f_{uv}\) in \(D\) corresponding to the edge \(uv\).
By construction, this vertex is not covered by any shortest path
between vertices in \(S_0\).
Furthermore, neither the shortest path from \(u\) to \(d\) nor the
shortest path from \(v\) to \(d\) covers \(f_{uv}\).
Additionally, \(f_{uv}\) is not reachable from any vertex in
\(V_{i'}\) for \(i' \in [\ell] \setminus \{i,j\}\).
Thus, \(f_{uv}\) cannot lie on any shortest path between a pair of
vertices in \(S'\), contradicting the assumption that \(S'\)
is a geodetic set of \(D\).
Therefore, \(X\) is an independent set in \(H\) containing exactly
one vertex from each \(V_i\), completing the proof.
\end{proof}

The conditional lower bound mentioned 
in Theorem~\ref{thm-kernel} follows
from \Cref{prop-eth-ind-set,lemma-ind-set-geodetic-eth}.

%% file: subcubic-solution-size-reduction.tex
\section{\texorpdfstring{$\W[2]$}{W[2]}-hardness for solution size on subcubic DAGs}
\label{sec-w2-hardness-sol-degree}

We restate the result proved in this section. 

\thmWhard*
The proof is adapted from the reduction from \textsc{Set Cover} to \gsfull in~\cite{AraujoA22}, that implies its \NP-hardness and $\W[2]$-hardness on DAGs (of unbounded maximum degree).
In fact, we find it notationally cleaner to first revisit the proof from
\textsc{Set Cover} to \gsfull mentioned in~\cite{AraujoA22}.
Then, we modify the reduced instance to ensure that it has bounded maximum degree. Given an instance $(U,\mathcal{F},k)$ of \textsc{Set Cover} with 
\(U=\{1,\ldots,n\}\), and \(\mathcal{F}=\{F_1,\ldots,F_m\},\)
we construct a directed acyclic graph $D$ as follows.
The vertex set of $D$ is \(V(D) = X \cup Y \cup \{s,t_1,t_2\}\),
where \(X=\{f_1,\ldots,f_m\}\) and  \(Y=\{u_1,\ldots,u_n\}\).
The arc set $A(D)$ contains the following arcs:
\((i)\) For every set $F_i \in\mathcal{F}$ and element $j \in U$, 
if $j \in F_i$ then add the arc \((f_i,u_j)\in A(D)\).
\((ii)\) For each set vertex $f_i \in X$, add arcs \(sf_i\),  and \(f_it_1\).
\((iii)\) For each vertex $u_j \in Y$, add the arc \(u_jt_2\).
\((iv)\) Finally add a single arc \(st_2\in A(D)\).
It is easy to see that $D$ is a DAG, since we orient all arcs consistently from 
\(\{s\}\) towards the sinks $\{t_1, t_2\}$. 

We construct a new digraph \(D'\) from the digraph \(D\) obtained in the last reduction 
by locally replacing high-degree vertices with directed binary trees. 
The transformation ensures: \(V(D) \subseteq V(D')\),
\(D'\) is a DAG, the maximum degree in \(D'\) is \(3\), and finally 
\((D,k+3)\) is a \yes-instance of \textsc{Geodetic Set} if and only if \((D',k+3)\) is.

Let \(\Delta=\Delta(D)\) denote the maximum total degree of \(D\).
We define a \emph{binary tree of height exactly \(\lceil\log \Delta\rceil\)} as a
rooted tree in which every internal vertex has at most two children,
every root-to-leaf path has length exactly \(\lceil\log \Delta\rceil\), and the tree has exactly \(L\le \Delta\)
leaves. All internal vertices have total degree at most~3.

\begin{itemize}

\item Delete the vertex \(t_1\) from \(D\) (it will be reintroduced later on).

\item Construct an outgoing binary tree \(T_s^{\mathrm{out}}\) with
\(|N^+(s)|\) leaves, whose arcs are oriented away from its root \(s^+\).
For every original arc \(sf_i\) where \(f_i\in X\), delete it and instead connect the 
corresponding leaf of \(T_s^{\mathrm{out}}\) to \(f_i\).  Add the arc \(ss^+\).

\item For each \(f_i\in X\), construct an outgoing binary tree
\(T_{f_i}^{\mathrm{out}}\) with \(|N^+(f_i)|\) leaves, oriented away from
its root \(f_i^{+}\).  For every original arc \(f_iu_j\), delete it.
Similarly, for each \(u_j\in Y\), construct an incoming binary tree
\(T_{u_j}^{\mathrm{in}}\) with \(|N^-(u_j)|\) leaves, whose arcs are
oriented toward its root \(u_j^{-}\).

For every original arc \(f_iu_j\) of \(D\), connect one leaf of
\(T_{f_i}^{\mathrm{out}}\) to one distinct leaf of \(T_{u_j}^{\mathrm{in}}\).
This replaces all arcs from \(X\) to \(Y\).

\item Construct an incoming binary tree \(T_{t_2}^{\mathrm{in}}\) whose arcs are
oriented toward its root \(t_2^{-}\).  
Attach every leaf of \(T_{u_j}^{\mathrm{out}}\) corresponding to an
original arc \(u_jt_2\) to a distinct leaf of \(T_{t_2}^{\mathrm{in}}\).  
Add the arc \(t_2^{-}t_2\).

\item Reintroduce vertex \(t_1\) by adding it to the graph.
For every \(u_j\in Y\), add the arc \(u_j^{-}t_1\).
Construct an incoming binary tree \(T_{t_1}^{\mathrm{in}}\) with leaves
corresponding to all arcs \((u_j^{-},t_1)\), and orient all arcs toward
its root \(t_1^{-}\).  
Replace each arc \(u_j^{-}t_1\) by an arc from \(u_j^{-}\) to a distinct
leaf of \(T_{t_1}^{\mathrm{in}}\), and add the arc \(t_1^{-}t_1\).
\end{itemize}

Let \(D'\) be the resulting digraph.
We remark that arc \(st_2\) has not been disturbed by the above modifications and it is present in \(D'\).
All added gadgets are directed trees with consistent orientation.
Since the original graph is a DAG, replacing arcs by tree paths cannot create a directed cycle. 
Hence \(D'\) is a DAG.
Finally, every internal node of every tree has total degree at most~3, and each
original vertex is incident to at most one incoming and two outgoing arcs
and hence the maximum degree of \(D\) is at most \(3\). We are now ready to prove Theorem~\ref{thm-W-hard}.

\begin{proof}[Proof of Theorem~\ref{thm-W-hard}]
To prove the theorem, it is sufficient to prove 
that both the reductions are safe.
We first argue that \((U,\mathcal{F},k)\) is a \yes-instance of \textsc{Set Cover} if and only if \((D,k+3)\) is  a \yes-instance of \textsc{Geodetic Set}.

In the forward direction, suppose there is a subfamily $\mathcal{F}'\subseteq\mathcal{F}$ 
with $|\mathcal{F}'|\le k$ and $\bigcup_{F\in\mathcal{F}'}F=U$.  
Let \(S = \{f_i : F_i\in\mathcal{F}'\}\cup\{s,t_1,t_2\}\).
Then $|S| \le k+3$. It remains to prove that $S$ is a geodetic set of $D$.
Note that any vertex, say \(f_i\) in \(X\), is on a shortest path from \(s\) to \(f_i\) to \(t_1\).
Indeed, every element $u_j\in Y$ lies on a directed path
\((f_i,u_j,t_2) \) where $f_i\in S$ because $j\in F_i$ for some $F_i\in\mathcal{F'}$.  

In the reverse direction, suppose $S\subseteq V(D)$ be a geodetic set with $|S|\le k+3$.  
Since $s$, $t_1$, and $t_2$ are respectively a source and sinks in $D$, 
they must belong to every geodetic set.  
Thus, at most $k$ vertices of $S$ can belong to $X \cup Y$.
If any $u_j\in Y$ is in $S$, we can replace it with $f_i \in X$ such that \((f_i, u_j)\) in \(A(D)\).
Note that even with this modifications, \(S\) is a geodetic set of \(D\).  
Hence, without loss of generality, assume $S\setminus\{s,t_1,t_2\} \subseteq X$.  
Let $\mathcal{F}'=\{F_i\mid f_i\in S\setminus\{s,t_1,t_2\}\}$.  
The argument above shows that for each $u_j \in Y$ there must be some $f_i\in S$ with 
$(f_i,u_j)\in A(D)$,  meaning $j\in F_i$.  
Therefore $\mathcal{F}'$ covers $U$, and $|\mathcal{F}'|\le k$.

This proves the correctness of the reduction.
Hence, \((U,\mathcal{F},k)\) is a \yes-instance of \textsc{Set Cover} if and only if \((D,k+3)\) is 
a \yes-instance of \textsc{Geodetic Set}.
In the remaining proof,  we claim that \((D,k+3)\) is a \yes-instance if and only if \((D',k+3)\) is.

In the forward direction, suppose that \(S\) is a geodetic set in \(D\).
As argued in the first stage of the reduction, it is safe to assume that
\(S\) contains \(s, t_1,\) and \(t_2\), since \(s\) is the unique source and
\(t_1, t_2\) are sinks in \(D\), and that \(S \setminus \{s, t_1, t_2\} \subseteq X\).
It is easy to verify that the shortest paths from \(s\) to \(t_1\) cover
all vertices in \(V(D')\) except those in \(Y\) and the incoming tree of
\(t_2\).
For the remaining vertices, consider an arbitrary vertex \(u_j \in Y\).
Since this vertex is covered in \(D\), there exists a vertex \(f_i \in S\)
such that \((f_i, u_j) \in A(D)\), and hence path \(f_i\) to \(u_j\) to \(t_2\)
covers \(u_j\).
It is straightforward to check that the path from \(f_i\) to \(t_2\) via
\(u_j\) is a shortest path in \(D'\).
Consequently, this path covers \(u_j\) as well as all vertices on the path
from the leaf of the incoming tree \(T_{t_2}^{\mathrm{in}}\) adjacent to \(u_j\) to \(t_2\).
As \(u_j\) was chosen arbitrarily, the same argument applies to every
vertex in \(Y\). Hence, all vertices in \(Y\) and all vertices of the incoming binary tree
\(T_{t_2}^{\mathrm{in}}\) are covered.

In the reverse direction, suppose that \(S'\) is a geodetic set of \(D'\)
with \(|S'| \le k+3\). We show that \((D,k+3)\) is a \yes-instance of \textsc{Geodetic Set}.
As \(s\) is the unique source and \(t_1,t_2\) are sinks in \(D'\), 
every geodetic set of \(D'\) must contain \(s,t_1,\) and \(t_2\).
Therefore, \(S' \setminus \{s,t_1,t_2\}\) contains at most \(k\) vertices.

We first observe that the shortest directed paths from \(s\) to \(t_1\)
cover all vertices of \(V(D')\) except those corresponding to vertices in
\(Y\) and the vertices of the incoming binary tree \(T_{t_2}^{\mathrm{in}}\).
Indeed, every vertex outside these sets lies on a directed path obtained
by replacing an original arc of the form \(sf_i\) or \(f_it_1\) by
its corresponding tree path, and all such paths have equal length.
Thus, these vertices are covered independently of the choice of
\(S' \setminus \{s,t_1,t_2\}\).

We now modify \(S'\) to obtain a geodetic set that intersects only the set
\(X\), besides \(s,t_1,t_2\).
If \(S'\) contains a vertex \(u_j \in Y\) or a vertex in the gadget
corresponding to \(u_j\), we replace it by some vertex \(f_i \in X\) such
that \(f_iu_j\) was an arc in the original graph \(D\).
Such a vertex \(f_i\) exists, since otherwise \(u_j\) would be unreachable
from \(s\), contradicting the fact that \(S'\) is geodetic.
This replacement preserves the geodetic property, because every shortest
path from \(f_i\) to \(t_2\) passes through \(u_j\) and the corresponding
vertices of the incoming tree of \(t_2\).
Finally, if there is any vertex of \(S'\) in out-tree of \(f_i\), then we 
replace it by \(f_i\). 
Note that this does not change the vertices in
\(Y\) and the vertices of the incoming binary tree \(T_{t_2}^{\mathrm{in}}\),
that were earlier covered by vertices in \(S'\).
After exhaustively applying this operation, we obtain a geodetic set
\(S\) with \(|S| \le k+3\) such that \(S \setminus \{s,t_1,t_2\} \subseteq X\).
It is easy to see that \(S\) covers all vertices of \(D\).
Finally, projecting \(S\) onto \(V(D)\) yields a geodetic set of \(D\) of
size at most \(k+3\). 

This implies \((D,k+3)\) is a \yes-instance if and only if \((D',k+3)\) is.
By pipelining the previous reduction, this implies that \textsc{Geodetic Set}
remains \W[2]-hard for solution size \(k\), even on DAGs with 
maximum degree 3.
\end{proof}

%% file: mfd+max-degree-reduction.tex
\section{\NP-hardness for subcubic DAGs with small reachability diameter}
\label{sec-np-hard-max-deg-mfd}

We restate the result proved in this section. 

\thmmaxdeg*

We prove the theorem by reducing from a variant of \textsc{3-SAT} where the number of occurrences of a variable in clauses is at most three. 

\problemdc{\textsc{3-Occ 3-SAT}}{A set $X$ of $n$ variables and a set $C$ of $m$ clauses such that any variable appears at most three times and each clause contains either two or three variables.}{Find a truth assignment $\varphi$ of the variables of $X$ such that it satisfies all clauses of $C$.}{1}

\begin{theorem}[\cite{tovey1984}, Theorem 2.1]
    \textsc{3-Occ 3-SAT} is \NP-complete.
\end{theorem}

We consider an instance $\mathcal I = (X,C)$ of \textsc{3-Occ 3-SAT} and construct a digraph $D =(V,A)$ based on $\mathcal I$. Without loss of generality, we order the set of variables $X$ such that $X=\{x_1, \dots, x_n\}$. We also order the set of clauses $C$ such that $C = \{c_1, \dots, c_m\}$.

Before describing the structure of $D$, we introduce \emph{link gadgets} that will be used throughout the construction of $D$ to ensure the maximum degree of $D$ remains low. Those gadgets are similar to the binary trees used in \Cref{sec-w2-hardness-sol-degree}. A \emph{simple link gadget} $L$ is a binary tree of height exactly $2$ with three leaves. Denote its vertex set as $V^L = \{m_1^L, m^L_2, m^L_3\}\cup \{a^L_1, a^L_2\}\cup \{u^L\}$, where the vertex $u^L$ is the root, the three vertices $m_i^L$ for $1\leq i\leq 3$ the leaves, and vertices $a^L_1$ and $a^L_2$ are the inner vertices of said binary tree. Furthermore, the orientation of the arcs is either from $u^L$ to vertices $m_i^L$ (in this case, the linking gadget is said to be \emph{dividing}) or from vertices $m_i^L$ to $u^L$ (and the linking gadget is said to be \emph{grouping}).
    
Finally, we also define an \emph{extended} version of both types of gadgets. To do so, subdivide the arcs respectively from $a_1^L$ and $a_2^L$ to $u^L$ by respectively adding the vertices $b_1^L$ and $b_2^L$ and orient the new arcs in the direction of the subdivided arcs. To distinguish the extended version of the gadgets from the original one, we will denote the latter with a superscript bar, so that the notation $L$ corresponds to a simple linking gadget, while $\bar L$ corresponds to an extended one. \Cref{fig-delta+mfd-linkgadget} depicts one simple and one extended link gadget, and introduces a compact representation for both types of link gadgets.

\begin{figure}
    \begin{subfigure}{.4\textwidth}
        \begin{tikzpicture}[
            vertex/.style={circle, fill,minimum size=2mm, inner sep=0pt},
            >={Stealth[scale=1.5]},
            scale=.8
        ]
        
        \node[vertex, label=$u^L$] (uL) {};
        \node[vertex, above right= .7cm of uL, label=$a_1^L$] (a1L) {};
        \node[vertex, below right= .7cm of uL, label=$a_2^L$] (a2L) {};
        \coordinate (xright) at ($(uL) + (2,.25)$);
        \node[vertex, label=$m_1^L$] (m1L) at ($(xright |- a1L)+(0,.75)$){};
        \node[vertex, label=$m_2^L$] (m2L) at (xright) {};
        \node[vertex, label=$m_3^L$] (m3L) at ($(xright |- a2L)$) {};
        
        \draw[->] (uL) -- (a1L);
        \draw[->] (uL) -- (a2L);
        \draw[->] (a1L) -- (m1L);
        \draw[->] (a1L) -- (m2L);
        \draw[->] (a2L) -- (m3L);
        
        \end{tikzpicture}
        \subcaption{A simple dividing link gadget $L$.}
        \label{fig-linking-simple}
    \end{subfigure}
    \hfill
    \begin{subfigure}{.4\textwidth}
    \centering
        \begin{tikzpicture}[
            vertex/.style={circle, fill,minimum size=2mm, inner sep=0pt},
            >={Stealth[scale=1.5]},
            scale=.8
        ]
        
        \node[vertex, label=above right:$u^{\bar L}$] (uL) {};
        \node[vertex, above left= .7cm of uL, label=$b_1^{\bar L}$] (b1L) {};
        \node[vertex, below left= .7cm of uL, label=$b_2^{\bar L}$] (b2L) {};
        \node[vertex, left=of b1L, label=$a_1^{\bar L}$] (a1L) {};
        \node[vertex, left=of b2L, label=$a_2^{\bar L}$] (a2L) {};
        \coordinate (xleft) at ($(uL) + (-3.5,.25)$);
        \node[vertex, label=$m_1^L$] (m1L) at ($(xleft |- a1L)+(0,.75)$){};
        \node[vertex, label=$m_2^L$] (m2L) at (xleft) {};
        \node[vertex, label=$m_3^L$] (m3L) at ($(xleft |- a2L)$) {};
        
        \draw[<-] (uL) -- (b1L);
        \draw[<-] (uL) -- (b2L);
        \draw[<-] (b1L) -- (a1L);
        \draw[<-] (b2L) -- (a2L);
        \draw[<-] (a1L) -- (m1L);
        \draw[<-] (a1L) -- (m2L);
        \draw[<-] (a2L) -- (m3L);
        
        \end{tikzpicture}
        \subcaption{An extended grouping link gadget $\bar L$.}
        \label{fig-linking-extended}
    \end{subfigure}

    \par\medskip
    
    \begin{subfigure}{\textwidth}
        \centering
        \begin{tikzpicture}[scale=.7, every label/.style={label distance=-3pt},
        >={Stealth[scale=1.5]}]
        
        \draw[rounded corners=12pt]
          (0,0) rectangle (3,3);
        
        \node[label=right:$u^L$] at (0,1.5) {$\bullet$};
        
        \node[label=left:$m_1^L$] at (3,2.2) {$\bullet$};
        \node[label=left:$m_2^L$] at (3,1.5) {$\bullet$};
        \node[label=left:$m_3^L$] at (3,0.8) {$\bullet$};
    
        \node at (1.5,1.5) {$L$};
        
        \draw[->, thick] (1.1,1.9) -- (1.9,1.9);
        
        \begin{scope}[xshift=7cm]
        
        \draw[rounded corners=12pt]
          (0,0) rectangle (3,3);
        
        \node[label=right:$m_1^{\bar L}$] at (0,2.2) {$\bullet$};
        \node[label=right:$m_2^{\bar L}$] at (0,1.5) {$\bullet$};
        \node[label=right:$m_3^{\bar L}$] at (0,0.8) {$\bullet$};
        
        \node[label=left:$u^{\bar L}$] at (3,1.5) {$\bullet$};
        
        \node[label=below:$b^{\bar L}_1$] at (2.2,3) {$\bullet$};
        
        \node[label=$b^{\bar L}_2$] at (2.2,0) {$\bullet$};
        
        \node at (1.5,1.5) {$\bar L$};
        
        \draw[->, thick] (1.1,1.9) -- (1.9,1.9);
        
        \end{scope}
        
        \end{tikzpicture}
        \subcaption{The corresponding compact representations of $L$ and $\bar L$. Note that the labeling of the boundary vertices of each gadget may be omitted in further representations.}
        \label{fig-linking-compact}
    \end{subfigure}
    \caption{Representations of link gadgets.}
    \label{fig-delta+mfd-linkgadget}
\end{figure}
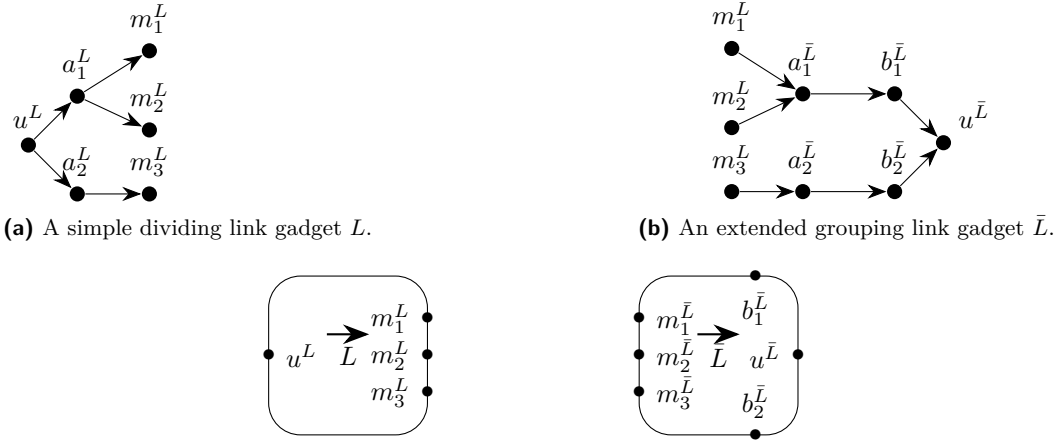

\begin{description}
    \item[Variable gadgets.] For each variable $x\in X$, add in $V$ a set $V^x$ of vertices such that $V^x = \{s^x, t_1^x, t_2^x, a^x\}\cup\{\mu^x_i\mid \mu\in \{\nu, \pi\},\;i\in[7]\}$. We also add:
    \begin{itemize}
        \item three arcs from $s^x$ respectively to $\pi^x_0$, $\nu^x_0$ and $t_2^x$;
        \item for all $\mu\in \{\pi, \nu\}$ and all $i\in[0,6]$, an arc from $\mu^x_i$ to $\mu^x_{i+1}$;
        \item for all $\mu\in \{\pi, \nu\}$, two arcs from $\mu_7^x$ respectively to $t_1^x$ and $a^x$;
        \item an arc from $a^x$ to $t_2^x$.
    \end{itemize}
    
    \item[Clause gadgets.] For any clause $c\in C$, add one extended grouping gadget $\bar L^c$ and a dividing gadget $L^c$, along with two vertices $t_1^c$ and $t_2^c$. Call $V^c$ the set of vertices of a clause gadget, such that $V^c = V^{L^c}\cup V^{\bar{L^c}}\cup\{t_1^c, t_2^c\}$. Add also the following arcs:
    \begin{itemize}
        \item two arcs respectively from $b_1^{\bar L^c}$ and $b_2^{\bar L^c}$ to $t_1^c$.
        \item an arc from $u^{\bar L^c}$ to $u^{L^c}$.
        \item three arcs respectively from $m^{L^c}_1,\;m^{L^c}_2$ and $m^{L^c}_3$ to $t^c_2$.
    \end{itemize}

    \item[Linking variables and clauses.] Consider a variable $x\in X$ and denote by $C(x)$ the set of clauses $\{c_i, c_j, c_k\}$ such that $x$ occurs in $c_i,\; c_j$ and $c_k$, and $i\leq j\leq k$. Note that the set $C(x)$ can also be of size 1 or 2, in which case we respectively consider that $i=j=k$ and $j=k$. Define also $C_0(x) = c_i$, $C_1(x) = c_j$ and $C_2(x) = c_k$. We consider similar notions for a clause $c\in C$, so that the set of variables appearing in $c$ is denoted by $X(c) = \{x_i, x_j,x_k\}$ with $i\leq j\leq k$. We also define $X_0(c) = x_i$, $X_1(c) = x_j$ and $X_2(c) = x_k$. We will add to $A$ some arcs by following the next procedure:
    \begin{itemize}
        \item for $1\leq p\leq |C(x)|$:
        \begin{itemize}
            \item Let $c = C_p(x)$ and $q$ be the smallest integer such that $X_q(c) = x$,
            \item if $x$ appears as a positive literal in $c$, then add the arcs from $\pi^x_{p}$ to $m^{L^c}_q$ and from $\pi^x_{3+p}$ to $m^{\bar L^c}_q$, otherwise add the arc from $\nu^x_{p}$ to $m^{L^c}_q$ and from $\nu^x_{3+p}$ to $m^{\bar L^c}_q$.
        \end{itemize}
    \end{itemize}
\end{description}

This concludes the construction of $D$. A partial representation of $D$ displaying the gadgets corresponding to a variable $x$ appearing in a clause $c$ can be found in \cref{fig-delta+mfd-reduction}. We first prove that the reachability diameter and the feedback vertex number of the obtained graph are constant. 

\begin{observation}\label{obs-delta+mfd-deg}
    The digraph $D$ has maximum degree equal to 3
    and reachability diameter equal to 14.
\end{observation}

\begin{proof}
    By the previous construction, no vertex is adjacent to more than three arcs. Furthermore, it can be checked that the longest shortest path in the graph is of size at most 13. Indeed, consider some variable $x$ adjacent to some clause $c$, and suppose that $C_3(x) = c$. Then, there exists a path from $s^x$ to $t_2^c$ going though the vertex $\pi_6^x$ which has length 14 (in fact, all such paths in $D$ have length 14).
\end{proof}

Recall \Cref{lem-sink-sources} which states that any extremal vertex of a given digraph is a part of any geodetic sets of said digraph. Consider the set $V_0$ of extremal vertices of $D$. Then we have $V_0 = \{s^x, t^x_1, t^x_2\mid x\in X\}\cup\{t^c_1, t^c_2\mid c\in C\}$, and the only sources of $D$ are the vertices $s^x$ for $x\in X$. \Cref{prop-delta+mfd-extremalcover} describes which vertices of $D$ are not covered by $V_0$.

\begin{proposition}\label{prop-delta+mfd-extremalcover}
    Let $V_0$ be the set of extremal vertices of $D$. It holds that
    $$I(V_0) = V\setminus\left(\{a^x \mid x\in X\}\cup\{u^{L^c}, u^{\bar L^c}, a^{L^c}_i \mid c\in C,\; i\in[2]\}\right).$$
\end{proposition}

\begin{proof}
    We consider all pairs of sources and vertices of $V_0$ to enumerate exhaustively the vertices that are covered by $V_0$. It can be verified that a source $s^x$ cannot reach any vertex of $V^{x'}$ if $x\not=x'$, as well as any vertex of $V^c$ if $x$ does not appear in $c$. This means that the only pairs to consider are of the form $(s^x, t^x_i)$ for some $x\in X$ and $(s^x, t_i^c)$ for some $c\in C$ containing some vertex $x$ as a litteral. It holds that: 
    \begin{enumerate}
        \item for all variables $x\in X$, the shortest paths between the source $s^x$ and the sink $t_1^x$ cover the set of vertices $\{\mu_i^x\mid \mu\in\{\pi, \nu\},\; i\in [7]\}$. $s^x$ is also adjacent to $t^x_2$, so no vertices are covered by this latter pair;\label{item1}
        \item for all $c\in C$, shortest paths between the sources $s^x$ for $x\in X(c)$ and sink $t^c_1$ cover all vertices of $V^{\bar L^c}$ except $u^{\bar L^c}$, and some vertices stated as covered in \cref{item1}. The shortest paths between the source $s^x$ for $x\in X(c)$ and $t^c_2$ cover vertices of $\{m^{L^c}_i\mid i\in [3]\}$ (and again some vertices stated as covered in \cref{item1}).
    \end{enumerate}

    We considered exhaustively all shortest paths between pairs of vertices of $V_0$, an the only vertices that are not covered by those paths are exactly $\{a^x \mid x\in X\}\cup\{u^{L^c}, u^{\bar L^c}, a^{L^c}_i \mid c\in C,\; i\in[2]\}$, thus concluding the proof.
\end{proof}

\begin{lemma}\label{lemma-delta+mfd-clausecovered}
    Let $c\in C$ and $v,w\in V$ such that $u^{\bar L^c}\in I(v,w)$. If $v\not\in \{s^x\}\cup \{\mu^x_i\mid \mu\in\{\pi, \nu\},\;i\in[3]\}$ for all variables $x$ contained in $c$, then for all vertices $u'\in  V^{L^c}\cup \{u^{\bar L^c}, s^c_2\}$, it holds that $u'\in I(v, t^c_2)$.
\end{lemma}
\begin{proof}
    Denote by $U^x$ the set $\{s^x\}\cup \{\mu^x_i\mid \mu\in\{\pi, \nu\},\;i\in[3]\}$. Suppose that $u^{\bar L^c}\in I(v,w)$ with $v\not\in U^x$ for all $x$ contained in $c$. Let $W = V^{L^c}\cup \{u^{\bar L^c}, s^c_2\}$. Since $v \not = s^x$, $v$ can only reach vertices of $W$ using paths going through $u^{\bar L^c}$. Furthermore, the vertices reachable from $u^{\bar L^c}$ are exactly the vertices in $W$. One can also verify that, by definition of $D$, $I(u^{\bar L^c}, t^c_2) = W$. Indeed, all paths from $u^{\bar L^c}$ to $t^c_2$ have length 4, and together span $W$. This means that a shortest path from $v$ to $u^{\bar L^c}$, combined with one of the shortest path between $u^{\bar L^c}$ and $t^c_2$ forms a shortest path from $v$ to $t^c_2$. All such combined path span $W$, thus ensuring that for all $u'\in W$, $u'\in I(v, t^c_2)$.
\end{proof}

\begin{observation}\label{obs-delta+mfd-nota}
    Let $S$ be a geodetic set of $D$ such that $a^x\in S$ for some $x\in X$. Let $v\in V^x\setminus(V^0\cup\{a^x\})$ Then, $(S\cup\{v\})\setminus \{a^x\}$ is a geodetic set of $D$
\end{observation}
\begin{proof}
    Since $a^x$ can only reach vertex $t^x_2$ and can only be reached by vertices from $V^x$, and by \Cref{prop-delta+mfd-extremalcover}, removing it from $S$ can only affect the coverage of $a^x$ itself. Furthermore, it is easy to see that any vertex $v$ from $V^x\setminus(V^0\cup\{a^x\})$ in has a unique path to $t^x_2$, and that this path covers $a^x$.
\end{proof}

We now prove the next two main lemmas of this section, that will justify the soundness of the reduction. We then end the section by proving \Cref{thm-maxdeg}.

\begin{lemma}\label{lemma-delta+mfd-sat2gs}
    If the instance $\mathcal I=(X,C)$ is a \yes-instance for \textsc{3-Occ 3-SAT}, then $D$ has a geodetic set of size at most $4n+2m$
\end{lemma}
\begin{proof}
    Let $\varphi$ by a satisfying assignment of $I$. Define a vertex set $S$ as follows. Add all extremal vertices to $S$, so that $V^0\subseteq S$. Furthermore, if $x$ is assigned positively in $\varphi$, then add vertex $\pi^x_4$ to $S$. Otherwise, add vertex $\nu_4^x$ to $S$. The constructed set $S$ is of size $|V^0| + n = 4n+2m$. Since $V^0\subseteq S$, all vertices in $I(V^0)$ are covered by $S$. Let us verify that all other vertices belong to $I(S)$. Recall that, by \Cref{prop-delta+mfd-extremalcover}, $V\setminus I(V^0) = \{a^x \mid x\in X\}\cup\{u^{L^c}, u^{\bar L^c}, a^{L^c}_i \mid c\in C,\; i\in[2]\}$.
    
    Let $x\in X$ and consider vertex $a^x$. By construction of $S$, either $\pi^x_4$ or $\nu^x_4$ belongs to $S$. Call this  selected vertex $v$. There exist a unique path from $v$ to $t^x_2$, and this path goes through $a^x$, thus $a^x$ belongs to $I(S)$.  

    Now, let $c$ be some clause of $C$. Since $\varphi$ is satisfying, there exist some variable $x\in X$ contained in $C$ that satisfies $c$. Let $p$ and $q$ be integers such that $C_p(x) = c$ and $X_q(c) = x$. If $x$ is present as a positive literal in $c$, then by construction of $S$, $\pi_4^x\in S$. Furthermore, by construction of $D$, the arc $\pi_{3+p}^x m_q^{\bar L^c}$ belongs to $D$. This means that there exists a unique path in $D$ between $\pi_{3+p}^x$ and $s_2^c$, that goes through $u^{\bar L^c}$. Then, by \Cref{lemma-delta+mfd-clausecovered}, all vertices of $V^{L^c}$ and $u^{\bar L^c}$ belong to $I(S)$. If $x$ is present as a negative literal in $c$, the same reasoning holds when replacing the letter $\pi$ by $\nu$. This concludes the proof since all vertices of $D$ belong to $I(S)$.
\end{proof}

\begin{lemma}\label{lemma-delta+mfd-gs2sat}
    If the digraph $D$ has a geodetic set of size at most $4n+2m$, then $\mathcal I$ is a \yes-instance for \textsc{3-Occ 3-SAT}.
\end{lemma}

\begin{proof}
    Suppose $D$ has a geodetic of size at most $4n+2m$, and call this set $S$. 
    By \Cref{lem-sink-sources}, we have that $V^0\subseteq S$. Furthermore, by \Cref{prop-delta+mfd-extremalcover}, $a^x \not\in I(V^0)$ for all $x\in X$. However, $a^x$ is only reachable from vertices of $V^x$, and the only vertex it can reach is $t^x_2$. This means that $S\cap \left(V^x\setminus V^0\right)\not=\emptyset$, otherwise $a^x$ could not be covered by $S$. Since there are exactly $n$ different variables in $X$, and $|V^0| = 3n+2m$, it means that $|S\cap \left(V^x\setminus V^0\right)| = 1$, and all non-extremal vertices of $S$ lie in variable gadgets. By \Cref{obs-delta+mfd-nota}, one can suppose that the non-extremal vertex of $S$ in $V^x$ is not $a^x$. Call this vertex $v(x)$.
    
    We define an assignment $\varphi$ over the variables of $X$ as follows: for any variable $x\in X$, $\varphi(x) = \TRUE$ if $v(x) = \pi^x_i$ for some $i\in [7]$, otherwise $\varphi(x) = \FALSE$. We claim that $\varphi$ satisfies all clauses of $C$. Indeed, consider some clause $c\in C$. Since $S$ is a geodetic set, there exist two vertices $s$ and $t$ of $S$ such that $u^{L^c}$ is lying on a shortest path from $s$ to $t$. Furthermore, by \Cref{prop-delta+mfd-extremalcover}, both $s$ and $t$ cannot belong to $V^0$. Since $u^{L^c}$ can only reach vertices of its clause gadget, $s$ has to be one non-extremal vertex of $S$. This means that there exist some variable $x$ for which the non-extremal solution vertex in its associated gadget covers $u^{L^c}$. Suppose without loss of generality that $s = \pi^x_i$ for some $i\in [7]$. There exists by construction of $D$ a path from $s$ to $t$ if and only if $x$ is contained as a positive literal in $c$, and thus the assignment of $x$ in $\varphi$ satisfies $c$. The same applies when $s = \nu^x_i$, and this holds for all $c\in C$, so $\varphi$ is a satisfying assignment.
\end{proof}

\begin{proof}[Proof of \Cref{thm-maxdeg}]
    \gsfull is known to be \NP-complete for general (di)graphs, it thus belongs to \NP. It cans also be verified that $D$ can be constructed in polynomial time from $\mathcal I$. Furthermore, by \Cref{lemma-delta+mfd-gs2sat,lemma-delta+mfd-sat2gs}, $\mathcal I$ is a \yes instance for \textsc{3-Occ 3-SAT} if and only if $D$ has a geodetic set of size less than $4n+2m$. Finally, \Cref{obs-delta+mfd-deg} confirms that the reduction holds for graphs of maximum degree 3 and reachability diameter 14.
\end{proof}
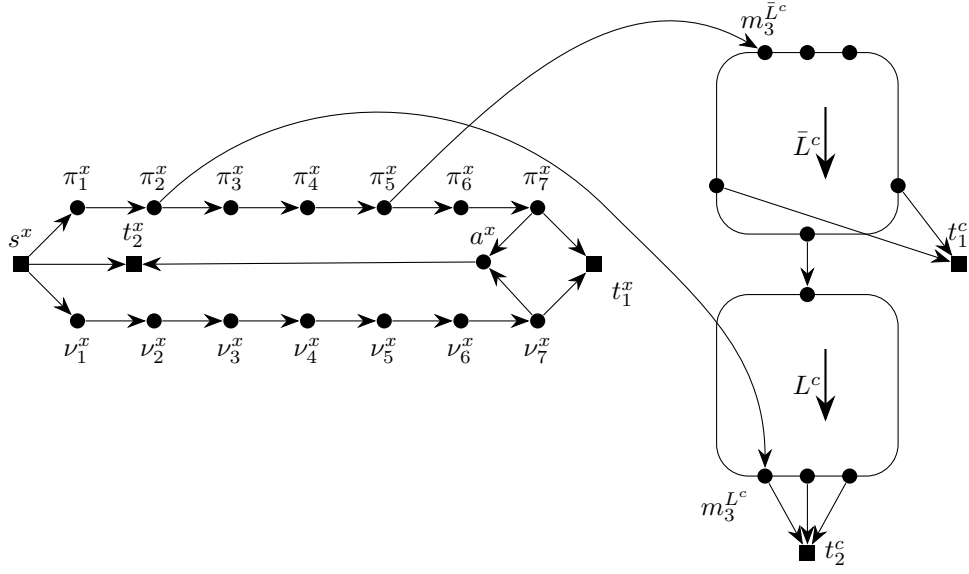
\begin{figure}
    \centering

    \begin{tikzpicture}[
    vertex/.style={circle, fill, minimum size=2mm, inner sep=0pt},
    >={Stealth[scale=1.5]}, scale=.8, node distance=.8cm]

        \node[vertex, label=$s^x$, rectangle] (s) {};
        \node [vertex, label=$\pi^x_1$, above right=of s] (p1) {};
        \node [vertex, label=below:$\nu^x_1$, below right=of s] (n1) {};
    
        \draw [->] (s) -- (p1);
        \draw[->] (s) -- (n1);
    
        \def\prev{1}
        \foreach \i in {2,...,7} {
            \node[vertex, label=$\pi^x_\i$, right=of p\prev] (p\i) {};
            \node[vertex, label=below:$\nu^x_\i$, right=of n\prev] (n\i) {};
            \draw[->] (p\prev) -- (p\i);
            \draw [->] (n\prev) -- (n\i);
            \xdef\prev{\i}
        }
    
        \node[vertex, below right=of p7, label=below right:$t_1^x$, rectangle] (tx1) {};
        \node [vertex, below right=of p1, label=$t^x_2$, rectangle] (tx2) {};
        \node [vertex, below left =of p7, label=$a^x$] (ax) {};
    
        \draw[->] (s) -- (tx2);
        \draw[->] (p7) -- (ax);
        \draw[->] (p7) -- (tx1);
        \draw[->] (n7) -- (ax);
        \draw[->] (n7) -- (tx1);
        \draw[->] (ax) -- (tx2);

        \begin{scope}[shift={(13,2)}, rotate=-90, shift={(-1.5,-1.5)}]

        \draw[rounded corners=12pt]
          (0,0) rectangle (3,3);
        
        \node[vertex] at (0,2.2) {};
        \node[vertex] at (0,1.5) {};
        \node[vertex, label=$m^{\bar L^c}_3$] (mt3) at (0,0.8) {};
        
        \node[vertex] (ut) at (3,1.5) {};
        
        \node[vertex] (b1) at (2.2,3) {};
        \node[vertex] (b2) at (2.2,0) {};
        
        \draw[->, thick] (0.9,1.8) -- (2.1,1.8);
        
        \end{scope}
        
        \node at (13,2) {$\bar L^c$};
        \node[vertex, label=$t_1^c$, rectangle] (t1c) at (15.5,0) {};

        \draw[->] (b1) -- (t1c);
        \draw[->] (b2) -- (t1c);
        
        \begin{scope}[shift={(13,-2)}, rotate=-90, shift={(-1.5, -1.5)}]
        
        \draw[rounded corners=12pt]
          (0,0) rectangle (3,3);
        
        \node[vertex] (ub) at (0,1.5) {};
        
        \node[vertex] (mbl) at (3,2.2) {};
        \node[vertex] (mbc) at (3,1.5) {};
        \node[vertex, label=below left:$m^{L^c}_3$] (mbr) at (3,0.8) {};
        
        \draw[->, thick] (0.9,1.8) -- (2.1,1.8);
        
        \end{scope}
        
        \node at (13,-2) {$L^c$};
        \node[vertex, rectangle, label=right:$t^c_2$, below=of mbc] (t2c) {};
        \draw[->] (mbl) -- (t2c);
        \draw[->] (mbc) -- (t2c);
        \draw[->] (mbr) -- (t2c);
        \draw[->] (ut) -- (ub);

        \coordinate[right=of p7] (temp);
        \draw[->] (p2) to[out=45, in=135] (temp) to[out=-45, in=90] (mbr);

        \draw[->] (p5) to[in=150] (mt3);
        
    \end{tikzpicture}

    \caption{The gadgets of $D$ corresponding to a variable $x\in X$ appearing as a positive literal in clause $c\in C$ such that $C_2(x) = c$ and $X_3(c) = x$. The sources and sinks of $D$ are represented with squares, and link gadgets are represented using their compact representation.}
    \label{fig-delta+mfd-reduction}
\end{figure}

%% file: conclusion.tex
\section{Conclusion and further work}\label{sec-conclusion}

In this work, we studied some structural parameters that, combined together, give an optimal (under \ETH) \FPT\ algorithm to solve \gsfull. In particular, we considered the influence of the maximum degree of a digraph, which can be seen as some measure of sparsity. The parameterized complexity of \gsfull has been studied recently for other sparse digraphs, namely ``tree-like'' directed (acyclic) graphs \cite{DBLP:conf/caldam/Foucaud26}. A natural extension of this work would be to consider the parameterized complexity of \gsfull with respect to other parameters capturing some notion of sparsity, or other classes of digraphs that are sparse. One could also study the existence of other kernels for \gsfull on directed graphs, especially for cases where the problem is known to be \FPT. In particular, to the best of our knowledge, no polynomial size kernel for \gsfull on digraphs is known.

%% file: references1.bib
@book{DBLP:books/sp/CyganFKLMPPS15,
  author    = {M. Cygan and F. V. Fomin and L. Kowalik and D. Lokshtanov and
               D. Marx and M. Pilipczuk and M. Pilipczuk and S. Saurabh},
  title     = {Parameterized Algorithms},
  publisher = {Springer},
  year      = {2015}
}

@inproceedings{floISAAC20,
  author    = {D. Chakraborty and S. Das and F. Foucaud and H. Gahlawat and
               D. Lajou and B. Roy},
  title     = {Algorithms and complexity for geodetic sets on planar and chordal graphs},
  booktitle = {ISAAC 2020},
  series    = {LIPIcs},
  volume    = {181},
  pages     = {7:1--7:15},
  year      = {2020}
}

@article{DBLP:journals/tcs/ChakrabortyGR23,
  author  = {D. Chakraborty and H. Gahlawat and B. Roy},
  title   = {Algorithms and complexity for geodetic sets on partial grids},
  journal = {Theoretical Computer Science},
  volume  = {979},
  pages   = {114217},
  year    = {2023}
}

@inproceedings{floCALDAM20,
  author    = {D. Chakraborty and F. Foucaud and H. Gahlawat and
               S. K. Ghosh and B. Roy},
  title     = {Hardness and Approximation for the Geodetic Set Problem in Some Graph Classes},
  booktitle = {CALDAM 2020},
  series    = {Lecture Notes in Computer Science},
  volume    = {12016},
  pages     = {102--115},
  year      = {2020}
}

@book{bookGC,
  author    = {I. M. Pelayo},
  title     = {Geodesic Convexity in Graphs},
  publisher = {Springer},
  year      = {2013}
}

@book{BOOKaraujo2025,
  author    = {J. Araújo and M. C. Dourado and F. Protti and R. M. Sampaio},
  title     = {Introduction to Graph Convexity: An Algorithmic Approach},
  publisher = {Springer},
  year      = {2025}
}

@article{KK22,
  author  = {L. Kellerhals and T. Koana},
  title   = {Parameterized Complexity of Geodetic Set},
  journal = {Journal of Graph Algorithms and Applications},
  volume  = {26},
  number  = {4},
  pages   = {401--419},
  year    = {2022}
}

@article{harary1993,
  author  = {F. Harary and E. Loukakis and C. Tsouros},
  title   = {The geodetic number of a graph},
  journal = {Mathematical and Computer Modelling},
  volume  = {17},
  number  = {11},
  pages   = {89--95},
  year    = {1993}
}

@article{dourado2010,
  author  = {M. C. Dourado and F. Protti and D. Rautenbach and J. L. Szwarcfiter},
  title   = {Some remarks on the geodetic number of a graph},
  journal = {Discrete Mathematics},
  volume  = {310},
  number  = {4},
  pages   = {832--837},
  year    = {2010}
}

@inproceedings{ekim2012,
  author    = {T. Ekim and A. Erey and P. Heggernes and P. van’t Hof and D. Meister},
  title     = {Computing minimum geodetic sets of proper interval graphs},
  booktitle = {LATIN 2012},
  series    = {Lecture Notes in Computer Science},
  volume    = {7256},
  pages     = {279--290},
  year      = {2012}
}

@article{mezzini2018,
  author  = {M. Mezzini},
  title   = {Polynomial time algorithm for computing a minimum geodetic set in outerplanar graphs},
  journal = {Theoretical Computer Science},
  volume  = {745},
  pages   = {63--74},
  year    = {2018}
}

@article{wellpart,
  author  = {J. Ahn and L. Jaffke and O. Kwon and P. T. Lima},
  title   = {Well-partitioned chordal graphs},
  journal = {Discrete Mathematics},
  volume  = {345},
  number  = {10},
  pages   = {112985},
  year    = {2022}
}

@inproceedings{DIT21,
  author    = {T. Davot and L. Isenmann and J. Thiebaut},
  title     = {On the Approximation Hardness of Geodetic Set and Its Variants},
  booktitle = {COCOON 2021},
  series    = {Lecture Notes in Computer Science},
  volume    = {13025},
  pages     = {76--88},
  year      = {2021}
}

@article{chartrand2000geodetic,
  author  = {G. Chartrand and P. Zhang},
  title   = {The geodetic number of an oriented graph},
  journal = {European Journal of Combinatorics},
  volume  = {21},
  number  = {2},
  pages   = {181--189},
  year    = {2000}
}

@article{lu2007geodetic,
  author  = {C. Lu},
  title   = {The geodetic numbers of graphs and digraphs},
  journal = {Science in China Series A: Mathematics},
  volume  = {50},
  number  = {8},
  pages   = {1163--1172},
  year    = {2007}
}

@article{chang2004geodetic,
  author  = {G. J. Chang and L. D. Tong and H. T. Wang},
  title   = {Geodetic spectra of graphs},
  journal = {European Journal of Combinatorics},
  volume  = {25},
  number  = {3},
  pages   = {383--391},
  year    = {2004}
}

@article{dong2009upper,
  author  = {L. Dong and C. Lu and X. Wang},
  title   = {The upper and lower geodetic numbers of graphs},
  journal = {Ars Combinatoria},
  volume  = {91},
  pages   = {401--409},
  year    = {2009}
}

@article{Chartrand2003,
  author  = {G. Chartrand and J. F. Fink and P. Zhang},
  title   = {The hull number of an oriented graph},
  journal = {International Journal of Mathematics and Mathematical Sciences},
  volume  = {2003},
  number  = {36},
  pages   = {2265--2275},
  year    = {2003}
}

@article{Farrugia05,
  author  = {A. Farrugia},
  title   = {Orientable convexity, geodetic and hull numbers in graphs},
  journal = {Discrete Applied Mathematics},
  volume  = {148},
  number  = {3},
  pages   = {256--262},
  year    = {2005}
}

@article{AraujoA22,
  author  = {J. Araújo and P. S. M. Arraes},
  title   = {Hull and geodetic numbers for some classes of oriented graphs},
  journal = {Discrete Applied Mathematics},
  volume  = {323},
  pages   = {14--27},
  year    = {2022}
}

@inproceedings{BDM24,
  author    = {B. Bergougnoux and O. Defrain and F. Mc Inerney},
  title     = {Enumerating Minimal Solution Sets for Metric Graph Problems},
  booktitle = {WG 2024},
  series    = {Lecture Notes in Computer Science},
  volume    = {14760},
  pages     = {50--64},
  year      = {2024}
}

@inproceedings{floICALP24,
  author    = {F. Foucaud and E. Galby and L. Khazaliya and S. Li and
               F. Mc Inerney and R. Sharma and P. Tale},
  title     = {Problems in NP Can Admit Double-Exponential Lower Bounds When Parameterized by Treewidth or Vertex Cover},
  booktitle = {ICALP 2024},
  series    = {LIPIcs},
  volume    = {297},
  pages     = {66:1--66:19},
  year      = {2024}
}

@inproceedings{floSTACS25,
  author    = {F. Foucaud and E. Galby and L. Khazaliya and S. Li and
               F. Mc Inerney and R. Sharma and P. Tale},
  title     = {Metric Dimension and Geodetic Set Parameterized by Vertex Cover},
  booktitle = {STACS 2025},
  series    = {LIPIcs},
  volume    = {327},
  pages     = {33:1--33:20},
  year      = {2025}
}

@inproceedings{T25,
  author    = {P. Tale},
  editor       = {A. Agrawal and
                  E. Jan van Leeuwen},
  title        = {Geodetic Set on Graphs of Constant Pathwidth and Feedback Vertex Set
                  Number},
  booktitle    = {{IPEC} 2025},
  series       = {LIPIcs},
  volume       = {358},
  pages        = {28:1--28:14},
  publisher    = {Schloss Dagstuhl - Leibniz-Zentrum f{\"{u}}r Informatik},
  year         = {2025},
}

@article{tovey1984,
  author  = {C. A. Tovey},
  title   = {A Simplified NP-complete Satisfiability Problem},
  journal = {Discrete Applied Mathematics},
  volume  = {8},
  number  = {1},
  pages   = {85--89},
  year    = {1984}
}

@inproceedings{DBLP:conf/caldam/Foucaud26,
  author       = {F. Foucaud and
                  N. Ghareghani and
                  L. Lorieau and
                  M. Mohammad Noori and
                  R. Parvini Oskuei and
                  P. Tale},
  editor       = {N. Misra and
                  A. Pandey},
  title        = {Algorithms and Hardness for Geodetic Set on Tree-Like Digraphs},
  booktitle    = {CALDAM 2026},
  series       = {Lecture Notes in Computer Science},
  volume       = {16445},
  pages        = {179--193},
  publisher    = {Springer},
  year         = {2026},
  bibsource    = {dblp computer science bibliography, https://dblp.org}
}

@article{DBLP:journals/iandc/ChakrabortyDFGL26,
  author       = {D. Chakraborty and
                  S. Das and
                  F. Foucaud and
                  H. Gahlawat and
                  D. Lajou},
  title        = {Algorithms and complexity for geodetic sets on interval and chordal
                  graphs},
  journal      = {Information and Computation},
  volume       = {311},
  pages        = {105456},
  year         = {2026},
  bibsource    = {dblp computer science bibliography, https://dblp.org}
}

@article{bus,
  author  = {C. Wang and Y. Song and G. Fan and H. Jin and L. Su and
             F. Zhang and X. Wang},
  title   = {Optimizing Cross-Line Dispatching for Minimum Electric Bus Fleet},
  journal = {IEEE Transactions on Mobile Computing},
  volume  = {22},
  number  = {4},
  pages   = {2307--2322},
  year    = {2023}
}

@article{DBLP:journals/algorithmica/FoucaudGPS21,
  author  = {F. Foucaud and B. Gras and A. Perez and F. Sikora},
  title   = {On the Complexity of Broadcast Domination and Multipacking in Digraphs},
  journal = {Algorithmica},
  volume  = {83},
  number  = {9},
  pages   = {2651--2677},
  year    = {2021}
}

@inproceedings{ETH1999,
  author    = {R. Impagliazzo and R. Paturi},
  title     = {Complexity of k-SAT},
  booktitle = {Proceedings of the Fourteenth Annual IEEE Conference on Computational Complexity},
  pages     = {237--240},
  year      = {1999}
}

@inproceedings{harris_et_al2024,
  author    = {D. G. Harris and N. S. Narayanaswamy},
  title     = {A Faster Algorithm for Vertex Cover Parameterized by Solution Size},
  booktitle = {STACS 2024},
  series    = {LIPIcs},
  volume    = {289},
  pages     = {40:1--40:18},
  year      = {2024}
}

@article{lampis2025,
author = {Lampis, M. and Melissinos, N. and Vasilakis, M.},
title = {Parameterized Max Min Feedback Vertex Set},
journal = {SIAM Journal on Discrete Mathematics},
volume = {39},
number = {3},
pages = {1587-1620},
year = {2025}
}

@inproceedings{DBLP:conf/sosa/HaeuplerJS26,
  author       = {B. Haeupler and
                  Y. Jiang and
                  T. Saranurak},
  editor       = {S. Assadi and
                  E. Rotenberg},
  title        = {Reducing Shortcut and Hopset Constructions to Shallow Graphs},
  booktitle    = {{SOSA} 2026},
  pages        = {385--393},
  publisher    = {{SIAM}},
  year         = {2026},
  bibsource    = {dblp computer science bibliography, https://dblp.org}
}

@article{bals2026revisitingdiameterdirectedgraphs,
      title={Revisiting Diameter in Directed Graphs}, 
      author={B. Bals and J. Blikstad and D. Dadush and Y. Nazari and J. Schmidt},
      year={2026},
      journal = {arXiv eprint},
      pages={2606.08217},
      note = {\url{https://arxiv.org/abs/2606.08217}}, 
}
